\documentclass[11pt]{article}

\usepackage[sectionbib]{natbib}
\usepackage{algorithm}
\usepackage{booktabs}
\usepackage{lineno}
\usepackage{graphicx}
\usepackage{subcaption}
\usepackage{algpseudocode}
\usepackage{color,xcolor}
\usepackage{subcaption}
\usepackage{geometry}
\usepackage{multirow}
\usepackage{mathtools}

\usepackage[colorlinks=true,linkcolor=blue,citecolor=blue]{hyperref}%

\usepackage{amsmath}
\usepackage{amssymb}
\usepackage{amsthm}

\usepackage[utf8]{inputenc} 
\usepackage{hyperref}       
\usepackage{url}            
\usepackage{booktabs}       
\usepackage{amsfonts}       
\usepackage{nicefrac}       
\usepackage{microtype}      

\usepackage{natbib}

\usepackage{array,epsfig,rotating}
\usepackage{tabularx}
\usepackage{graphicx}
\usepackage{verbatim}
\usepackage{arydshln}
\usepackage{authblk}

\newtheorem{theorem}{Theorem}
\newtheorem{lemma}{Lemma}

\newtheorem{prop}{Proposition}

\newtheorem{assumption}{Assumption}

\allowdisplaybreaks

\def\hat{\widehat}
\def\tilde{\widetilde}

\newcommand{\bbb}{\boldsymbol \beta}

\newcommand{\bO}{\mathbf O}
\newcommand{\bo}{\boldsymbol o}

\newcommand{\bg}{\boldsymbol g}

\newcommand{\bx}{\boldsymbol x}

\newcommand{\be}{\boldsymbol e}

\newcommand{\PP}{\mathbb P}

\newcommand{\bX}{\mathbf X}

\newcommand{\bD}{\mathbf D}

\newcommand{\bS}{\mathbf S}
\newcommand{\bH}{\mathbf H}
\newcommand{\bV}{\mathbf V}

\newcommand{\bI}{\mathbf I}

\newcommand{\ME}{\mathbb E}
\newcommand{\MP}{\mathbb P}

\newcommand\independent{\protect\mathpalette{\protect\independenT}{\perp}}
\def\independenT#1#2{\mathrel{\rlap{$#1#2$}\mkern2mu{#1#2}}}

\DeclareMathOperator*{\argmin}{arg\,min}

\def\spacingset#1{\renewcommand{\baselinestretch}%
{#1}\small\normalsize} \spacingset{1}
\spacingset{1}

\usepackage{float}

\begin{document}

\newcommand{\blind}{0}

\newcommand{\tit}{Nonparametric Estimation of Local Treatment \\ Effects with Continuous Instruments} 

\if0\blind

{\title{\tit\thanks{The authors declare no conflicts. Research in this article was supported by the National Library of Medicine \#1R01LM013361-01A1 and NSF CAREER grant 2047444. All statements in this report, including its findings and conclusions, are solely those of the authors.}}
\author{Zhenghao Zeng\thanks{Stanford University, Email: zhzeng@stanford.edu}, Alexander W. Levis\thanks{Assistant Professor, University of Pennsylvania,, Email: alexander.levis@pennmedicine.upenn.edu}, JungHo Lee\thanks{Carnegie Mellon University, Email: junghol@andrew.cmu.edu}, \\ Edward H. Kennedy\thanks{Associate Professor, Carnegie Mellon University, Email: edward@stat.cmu.edu}, Luke Keele\thanks{Associate Professor, University of Pennsylvania, Email: luke.keele@gmail.com, corresponding author}
}

\date{}

\maketitle
}\fi

\if1\blind
\title{\bf \tit}
\maketitle
\fi

\maketitle

\begin{abstract}
Instrumental variable methods are widely used to address unmeasured confounding, yet much of the existing literature has focused on the binary instrument setting. Extensions to continuous instruments often impose strong parametric assumptions for identification and estimation, which can be difficult to justify and may limit their applicability in complex real-world settings. In this work, we develop theory and methods for nonparametric estimation of treatment effects with a continuous instrumental variable. We introduce an estimand that, under a monotonicity assumption, quantifies the treatment effect among the maximal complier class, generalizing the local average treatment effect framework to continuous instruments. Considering this estimand and the local instrumental variable curve, we draw connections to the dose-response function and its derivative, and propose doubly robust estimation methods. We establish convergence rates and conditions for asymptotic normality, providing valuable insights into the role of nuisance function estimation when the instrument is continuous. Additionally, we present practical procedures for bandwidth selection and variance estimation. Through extensive simulations, we demonstrate the advantages of the proposed nonparametric estimators. Finally, we apply our methods to data where excess travel time is an instrument for patients’ likelihood of receiving care at specialized health care facilities. We use this instrument to estimate the effect of delivering at low-quality neonatal intensive care units (NICUs) on infant mortality.
\end{abstract}

\noindent%

\thispagestyle{empty}

\clearpage

\spacingset{1.1}

\section{Introduction}
The method of instrumental variables (IVs), originally developed in econometrics in the 1920s, provides a powerful framework for drawing causal inferences in the presence of unobserved confounders. This approach relies on identifying an instrumental variable—a variable that is associated with the treatment of interest but affects outcomes only through its impact on treatment assignment. While the treatment itself may be confounded, IVs remain unconfounded by design or domain knowledge, allowing researchers to identify and estimate different forms of causal relationships, despite unmeasured confounding. Although IV methods were historically proposed in econometrics, the past three decades have witnessed increasing interest from the statistical community. Building upon the potential outcome framework, foundational works have extended IV methods to randomized experiments with noncompliance, relaxed treatment effect homogeneity assumptions, introduced estimation of treatment effects among compliers, and developed partial identification results \citep{robins1989analysis, robins1994correcting, angrist1996identification, imbens1994, Manski:1990, Balke:1997}. \\



Most research on IVs has focused on the canonical scenario where the instrument is binary. However, in many applications, IVs are continuous or nearly continuous measures. In such cases, the method of two-stage least squares (TSLS) method is often used for estimation. However, TSLS relies on parametric assumptions for identification and assumes constant treatment effects \citep{okui2012doubly}, which are often restrictive and unrealistic in practice. Recent research has introduced estimation methods that are far more flexible. These estimation methods are designed for continuous IVs but also incorporate doubly robust adjustments and allow for heterogeneous treatment effects \citep{tan2010marginal, kennedy2019robust, mauro2020instrumental,robins2001comment,vanderlaan2003unified}. One strand of this research has focused on identification and estimation of the local IV (LIV) curve \citep{heckman1997instrumental,heckman1999local,heckman2005structural,glickman2000derivation,vytlacil2002independence,kennedy2019robust}. The LIV framework invokes a generalization of the monotonicity assumption from the binary IV case for continuous IVs.  When the IV is binary, the monotonicity assumption stipulates that no units defy the encouragement of the instrument to receive treatment \citep{imbens1994}. Under monotonicity, the target causal estimand is (typically) the treatment effect specifically for the subgroup of compliers---those who follow the encouragement of the instrument. The LIV framework generalizes this monotonicity assumption to the continuous IV setting: under LIV, if the treatment is binary and monotone with respect to the IV, each unit has a latent threshold such that treatment is taken if and only if the IV exceeds that threshold. In this case, one can identify and estimate the so-called LIV curve which captures the treatment effect within subgroups with specific threshold values across the range of the IV. Early LIV estimators relied on restrictive parametric models \citep{basu2007use,carneiro2011estimating}, however, more recent work has developed semiparametric estimation methods that relax key parametric modeling assumptions. Notably, \cite{kennedy2019robust} introduced an approach that projects the LIV curve onto a parametric working model, ensuring that even if the model is misspecified, the estimand remains interpretable as the best approximation of the LIV curve within the chosen model class. However, when the working model is misspecified, the estimated projection may still lead to substantial estimation error, highlighting the need for more flexible, fully nonparametric approaches. \\

In this paper, we develop nonparametric estimation methods for applications with a continuous IV.  First, we develop fully nonparametric estimators for the LIV curve, addressing the challenges posed by the ratio-of-derivative structure. We note that the numerator and denominator in this ratio share the same structure as the derivative of the usual dose-response function. Leveraging this insight, we propose two doubly robust estimation methods for the LIV curve that rely on derivative estimates of two dose-response functions. Both methods target a smooth approximation of the derivative, and we provide a unifying framework for smooth dose-response estimation generalizing the underlying ideas. We also derive practical variance estimators and outline a cross-validation approach for tuning parameter selection. In summary, our methods enable flexible and efficient estimation using nonparametric machine learning techniques that avoid model misspecification while still allowing for valid statistical inference. \\

Second, we introduce two new estimands for continuous IVs that are bounded. The first we refer to as the maximal complier class probability, which represents the proportion of individuals in the study population whose treatment could be influenced as the instrument varies from its minimum to maximum value. The other estimand measures the treatment effect within this maximal complier class. While its structure mirrors the well-known local average treatment effect (LATE), it cannot be estimated at parametric rates due to the continuous nature of the instrument. We establish a connection between this estimand and the dose-response function, and propose a doubly robust method for its estimation. We then conduct a series of simulation studies to demonstrate the advantages of the proposed methods. Finally, we apply our methods to a well-known empirical study that uses geographic distance as an instrument for access to higher-quality medical care. \\

Our paper is organized as follows: Section \ref{sec:preliminary} introduces the problem setup, causal assumptions, and estimands of interest in the continuous IV setting, including the LIV curve and treatment effects among the maximal complier class. We highlight the connection between these estimands and the dose-response function and its derivative. A framework for doubly robust estimation of the dose-response function is also provided. In Section \ref{sec:dose-res-boundary}, we consider estimation of the dose-response function at the boundary, which is then used for estimating treatment effects within the maximal complier class. Our approach extends the local linear estimator in \cite{kennedy2017non} to a local polynomial estimator, allowing a better fit to the local curvature. Since the LIV curve can be expressed as the ratio of derivatives of two dose-response functions, we introduce two doubly robust methods for estimating the derivative of the dose-response function in Sections \ref{sec:loc_poly} and \ref{sec:smoothing}. The novel theoretical results for estimating the dose-response function and its derivative provide valuable insights into how nuisance function estimation influences the final estimation rate. In Section~\ref{sec:sims}, we study the finite-sample performance of our methods in simulated data. In Section~\ref{sec:app}, we illustrate the practical application of our methods with an empirical example. Additional results, including a practical bandwidth selection method, additional simulation studies, and technical proofs, are provided in the Appendix.

\section{Preliminaries}
\label{sec:preliminary}

In this section, we introduce notation and review the identification conditions for causal effects in the continuous instrumental variable setting. Based on these causal assumptions, we define causal estimands of interest, discuss their interpretation, and lay out corresponding identification results. 

\subsection{Setup \& Notation}
\label{sec:setup}

Suppose we observe $n$ i.i.d. observations $\{\bO_i=(\bX_i,Z_i, A_i, Y_i), 1\leq i \leq n\}$ with a generic observation written $\bO=(\bX,Z, A, Y)$, where $\boldsymbol{X} \in \mathcal{X} \subseteq \mathbb{R}^d$ is a vector of covariates, $Z \in \mathbb{R}$ is a continuous instrument, $A \in \{0,1\}$ is a binary exposure variable, and $Y \in \mathbb{R}$ is a real-valued outcome of interest. Let $\mathcal{O}=\mathcal{X} \times \mathcal{Z} \times \mathcal{A} \times \mathcal{Y}$ denote the support of $\bO=(\bX,Z, A, Y)$ and $\mathcal{Z}_0$ the set of instrument values of interest. We rely on the potential outcome framework \citep{rubin1974estimating} to define causal effects. Specifically, let $A^z$ and $Y^z$ denote the counterfactual exposure and outcome values, had the instrument been set to $Z = z$. We also define $Y^a$ and $Y^{za}$ as the potential outcomes under interventions setting $A=a$ and both $A=a$ and $Z=z$, respectively. After reducing the problem of estimating causal effects in the continuous instrument setting to estimating quantities related to the dose-response function, we use 
$Z$ to denote the treatment in Section~\ref{sec:dose-res-boundary}--\ref{sec:smoothing}. \\

For distribution $\PP$ of  $\bO$ and a $\PP$-integrable function $\eta(\bO)$, we define $\PP[ \eta(\bO)] = \int \eta(\bo) \, d\PP(\bo)$, averaging over the randomness of $\bO$ while conditioning on $\eta$ when it is random. If $\eta$ is $\PP$-square-integrable, we denote its $L_2(\PP)$-norm as $\|\eta\|_2 = \sqrt{\int \eta^2(\bo) d \PP(\bo)}$.  For $n$ i.i.d. copies of $\bO$, we denote by $\mathbb{P}_n$ the empirical distribution and $\mathbb{P}_n [\eta(\bO)]$ the sample average $n^{-1}\sum_{i=1}^n \eta(\bO_i)$.  \\


Next, we introduce notation for three nuisance functions. These nuisance functions are necessary for estimation but are not of direct interest in themselves. First, let $\pi(Z \mid \bX)$ denote the conditional density of the instrumental variable $Z$ given the covariates $\bX$, also known as the instrument propensity score. We also denote the marginal density of $Z$ as $f(Z)$. Second, define $\lambda(\bX,Z):=\ME[A\mid \bX, Z]$,  representing the conditional mean of the treatment $A$ given the instrumental variable $Z$ and covariates $\bX$. Finally, we let $\mu(\bX,Z):=\ME[Y\mid \bX, Z]$, which is the the conditional mean of the outcome $Y$ given the instrumental variable $Z$ and covariates $\bX$. We define the estimation errors of $\pi$ and $\mu$ based on a training set $D$ and bandwidth $h > 0$ as follows:
\[
    \begin{aligned}
      r_n(z_0) :=&\, \sup_{z \in \mathcal{Z}, |z-z_0|\leq h} \sqrt{\mathbb{E}_{\bX} \left[ \mathbb{E}_D \left(\hat{\pi}(z \mid \bX) - \pi(z \mid \bX)\right)^2 \right]} ,  \\
      s_n(z_0) :=&\, \sup_{z \in \mathcal{Z}, |z-z_0|\leq h} \sqrt{\mathbb{E}_{\bX} \left[ \mathbb{E}_D \left(\hat{\mu}(\bX,z) - \mu(\bX,z)\right)^2 \right]},
    \end{aligned} 
\]
which will be useful in characterizing how the estimation error depends on nuisance function estimation. Note that $r_n(z_0)$ and $s_n(z_0)$ measure the average estimation error over $\bX$, uniformly within an $h$-radius neighborhood centered at the target point $z_0$.
We often illustrate our results under the assumption that the nuisance functions are smooth. Mathematically, we say a function $f$ is $s$-smooth if it is $\lfloor s\rfloor$ times continuously differentiable with derivatives up to order $\lfloor s\rfloor$ bounded by some constant $L>0$ and $\lfloor s\rfloor$-order derivatives Hölder continuous, i.e.
\[
    \left|D^{\boldsymbol{\beta}} f(\bx)-D^{\boldsymbol{\beta}} f\left(\bx^{\prime}\right)\right| \leq L\left\|\bx-\bx^{\prime}\right\|_2^{s-\lfloor s\rfloor}
\]
for all $\boldsymbol{\beta} = (\beta_1,\dots, \beta_d)$ with $\sum_{i} \beta_i = \lfloor s\rfloor$, where $D^{\boldsymbol{\beta}}=\frac{\partial^{\boldsymbol{\beta}}}{\partial x_1^{\beta_1} \ldots \partial x_d^{\beta_d}}$ is the differential operator. \\

Finally, our work utilizes kernel-based estimators, so we introduce the necessary notation for kernel regression. Given a symmetric kernel function $K: \mathbb{R} \mapsto \mathbb{R}$ and a bandwidth parameter $h>0$, the localized kernel is defined as $K_h(z) = K(z/h)/h$. To capture the local curvature of target functions, we rely on high-order kernels or polynomial bases. We say a kernel $K$ is a $\ell$-th order kernel, for a positive integer $\ell$, if it satisfies $\int  K(u) d u = 1$ and 
    \[
    \begin{aligned}
         \int u^{j} K(u) d u = 0, \, 1\leq j \leq \ell, \ , \ \int |u|^{\ell} |K(u)| d u < \infty.
    \end{aligned}
    \]
    We denote the (rescaled) $p$-th order polynomial basis as $\bg_h(z) = (1,z/h, \dots, z^p/h^p)^{\top}$. \\

\subsection{Identification Assumptions}
\label{sec:causal-assumptions}

Next, we outline the assumptions necessary for identifying causal effects in the continuous IV design. First, we briefly review a set of assumptions that are standard in the instrumental variables literature \citep{angrist1996identification}:
\label{sec:assump}
\begin{assumption}[Consistency] \label{ass:consistency}
    $A = A^Z$ and $Y = Y^{ZA}$ almost surely.
\end{assumption}

\begin{assumption}[Positivity] \label{ass:positivity}
   $\pi(z \mid \bX)>0$ almost surely for $z \in \mathcal{Z}$.
\end{assumption}

\begin{assumption}[Unconfoundedness] \label{ass:UC}
    $Z \independent (A^z, Y^z) \mid \boldsymbol{X}$.
\end{assumption}

\begin{assumption}[Exclusion Restriction] \label{ass:ER}
    $Y^{za} = Y^a$ almost surely, for all $z \in \mathcal{Z}, a \in \mathcal{A}$.
\end{assumption}
 
Assumption~\ref{ass:consistency} says interventions on $Z$ and $A$ are uniquely defined and unaffected by other units’ interventions (i.e., there is no interference between subjects). Assumption~\ref{ass:positivity} implies that each unit has some chance of receiving each level of the instrument, regardless of covariate values. Assumption~\ref{ass:UC} states that conditional on measured covariates $\bX$, the instrument assignment is as-if randomized. The exclusion restriction implies that the effect of $Z$ on $A$ operates solely through $A$, meaning $Z$ has no direct effect on $Y$. See \citet{hernan2006instruments} and \citet{imbens2014instrumental} for detailed discussions, and \citet{baiocchi2014} for a broader introduction to the IV assumptions.
\\ 

For continuous IVs, Assumption~\ref{ass:positivity} requires additional consideration and discussion. With a binary instrument, positivity means that each subject in the population has a positive probability of receiving both possible instrument values. However, when $Z$ is multi-valued or continuous, positivity implies that each subject must have a positive conditional probability (or density) of receiving any $z \in \mathcal{Z}$, given their covariates. This requirement may be unrealistic if certain units in the data have no chance of being exposed to instrument values far from those they actually received. For approaches that relax the positivity assumption with continuous instruments see \citet{rakshit2024local}. \\

These assumptions are necessary but not sufficient for point identification. In the binary IV setting, monotonicity (i.e., the absence of defiers) is often invoked as an additional assumption that enables point identification of causal effects among the population of compliers \citep{imbens1994, imbens2014instrumental}. Generalizations of this monotonicity remain critical for identifying causal effects with a continuous IV, and we employ a version used in~\citet{kennedy2019robust}:

\begin{assumption}[Monotonicity] \label{ass:mono}
   If $z' > z$ then $A^{z'} \geq A^z$ almost surely.
\end{assumption}
\noindent This monotonicity assumption stipulates that higher values of the instrument can either encourage otherwise unexposed units to be exposed to treatment or have no effect at all. This implies that higher instrument values cannot discourage treatment exposure compared with lower values and there do not exist defiers in the population. It is important to note that \citet{glickman2000derivation} and \citet{vytlacil2002independence} demonstrated that this continuous version of the monotonicity assumption can equivalently be expressed as the following latent threshold model:
\begin{assumption}[Latent Threshold] \label{ass:thresh}
  $A^z = 1(z \geq T)$, for all $z \in \mathcal{Z}$, where $T \in [-\infty, \infty]$ is an unobserved random threshold. 
\end{assumption}
\noindent Assumption~\ref{ass:thresh} implies that each complier has a threshold instrument value---denoted $T$---above which they are exposed to the treatment. Large values of $T$ imply that it requires higher instrument values to encourage treatment exposure, suggesting that such units are inherently less inclined to receive treatment. \\

When the instrument is binary, under the monotonicity assumption, we can classify units into three principal strata: never-takers, always-takers, and compliers. In the continuous IV setting, $T$ defines these principal strata as follows:
\[
T= \begin{cases}-\infty & \text { if } A^z=1 \text { for all } z \text { (always-takers), } \\ \inf \left\{z: A^z=1\right\} & \text { if } A^{z^{\prime}}>A^z \text { for some } z^{\prime}>z \text { (compliers), } \\ \infty & \text { if } A^z=0 \text { for all } z \text { (never-takers). }\end{cases}
\]
It is straightforward to see that Assumption~\ref{ass:thresh} implies Assumption~\ref{ass:mono}; conversely, under Assumption~\ref{ass:mono} the above display can be seen as a definition of $T$ which satisfies Assumption~\ref{ass:thresh}. Readers are referred to \cite{vytlacil2002independence} for additional discussion on monotonicity and latent index models. \\

Finally, we require the following regularity condition for the latent threshold $T$.

\begin{assumption}[Instrumentation] \label{ass:instru}
    The latent threshold $T$ is continuously distributed with a positive density on the set of instrument values of interest $\mathcal{Z}_0$:
    \[
    p(z_0):=\lim_{h \rightarrow 0} \frac{\MP(T\leq z_0 +h)-\MP(T\leq z_0 )}{h} >0, z_0 \in \mathcal{Z}_0.
    \]
\end{assumption}
\noindent The instrumentation Assumption \ref{ass:instru} implies that there are some units who would be exposed to the treatment when the instrument reaches $Z=z_0$.  This condition is analogous to the relevance assumption in the canonical IV design. That is, the instrument must encourage some units to be exposed to treatment. As in the binary IV case, estimation challenges may arise if the instrument is weak, i.e., if it has a nonzero but minimal effect on exposure. We will see in the next section, the density of $T$ can be identified and estimated from the data, allowing for an assessment of the strength of the continuous IV. In this work, we do not consider extensions for scenarios with weak instruments (in the sense that Assumption \ref{ass:instru} is violated).

\subsection{Target Causal Estimands}
\label{sec:estimand}

\subsubsection{Local Instrumental Variable Curve}
\label{sec:liv-estimand}

The first estimand of interest is the marginalized LIV curve, which \citet{kennedy2019robust} defined as 
\begin{equation}\label{eq:LIV}
   \gamma(z_0) = \ME[Y^{a=1} - Y^{a=0}\mid T=z_0]. 
\end{equation}
The LIV curve is the causal effect among those who would be treated precisely when the instrument reaches or exceeds $Z=z_0$, but would not be exposed at lower values. Early research focused on a version of the LIV curve that is fully conditional on $\bX$ \citep{heckman1997instrumental,heckman1999local,heckman2005structural}. Here, we focus on a marginal version of the LIV curve averaged over any non-effect modifiers in $\bX$. Note that the LIV curve differs from the more conventional IV causal effect known as LATE. The LATE with a continuous instrument is defined, for any pair $z, z' \in \mathcal{Z}$, as:
\begin{equation}\label{eq:late-continuous}
    \text{LATE}(z,z')=\mathbb{E}\left(Y^{a=1}-Y^{a=0} \mid A^{z}>A^{z'}\right),
\end{equation}
which represents the effect among those who would take the treatment at $Z=z$ but not at $Z=z'$. 
See Appendix~\ref{appendix:mte} for a detailed comparison to the marginal treatment effect framework of \citet{heckman1999local, heckman2001policy, heckman2005structural}. \\

Under Assumptions \ref{ass:consistency}--\ref{ass:instru} and assuming $\gamma$ is a continuous function, \citet{kennedy2019robust} showed that the LIV curve and the density of the latent threshold $T$ can be identified as
\begin{equation}\label{eq:liv-identification}
    \gamma(z_0)=\left.\frac{\frac{\partial}{\partial z} \mathbb{E}\{\mathbb{E}(Y \mid \mathbf{X}, Z=z) \}}{\frac{\partial}{\partial z} \mathbb{E}\{\mathbb{E}(A \mid \mathbf{X}, Z=z) \}}\right|_{z=z_0}=\left.\frac{\frac{\partial}{\partial z} \mathbb{E}[\mu(\bX,z) ]}{\frac{\partial}{\partial z} \mathbb{E}[\lambda(\bX,z)]}\right|_{z=z_0}
\end{equation}
\begin{equation}
\label{eq:liv-identification-density}
    \lim_{h \rightarrow 0} \frac{\MP(T \leq z_0+h)-\MP(T \leq z_0)}{h} = \frac{\partial}{\partial z} \mathbb{E}\{\mathbb{E}(A \mid \mathbf{X}, Z=z) \} |_{z=z_0} = \frac{\partial}{\partial z} \mathbb{E}[\lambda(\bX,z)] |_{z=z_0}.
\end{equation}
The identification proof closely follows the approach used when $Z$ is binary. We should also note that the LIV curve is only defined for finite $z_0 \in \mathcal{Z}_0$, and we cannot identify effects for always-takers ($T=-\infty$) and never-takers ($T=+\infty$). Critically, the ratio-of-derivatives structure of the LIV curve makes nonparametric estimation particularly challenging. \citet{kennedy2019robust} assumed a parametric working model for $\gamma(z_0)$ and developed doubly robust methods for the parameters that minimize the weighted distance between $\gamma(z_0)$ and the working model. \\

Here, we develop a  nonparametric estimator for $\gamma(z_0)$ by separately estimating the derivatives on the numerator and denominator in Equation~\eqref{eq:liv-identification}. Specifically, the numerator $\theta(z_0):= \frac{\partial}{\partial z} \mathbb{E}\{\mathbb{E}(Y \mid \mathbf{X}, Z=z) \}|_{z=z_0}$ has the same structure as the derivative of the ``dose-response curve'' in \citet{kennedy2017non}. That is, under Assumptions~\ref{ass:consistency}--\ref{ass:UC}, $\mathbb{E}\{\mathbb{E}(Y \mid \mathbf{X}, Z=z_0)\} \equiv \mathbb{E}(Y^{z_0})$ can be interpreted as the causal effect of setting the instrument $Z$ to the ``dose'' $z \in \mathcal{Z}$ on the outcome $Y$. Thus we use the term \textit{dose-response curves} to refer to the following functions:
\[
\tau(z_0)\coloneqq \ME[\ME(Y \mid \bX, Z=z_0)] \text{ and }\delta(z_0) \coloneqq \ME[\ME(A \mid \bX, Z=z_0)], \, z_0 \in \mathcal{Z}_0.
\]
The term $\delta(z_0)$ can be similarly interpreted as $\mathbb{E}(A^{z_0})$ under Assumptions~\ref{ass:consistency}--\ref{ass:UC}. Below, without loss of generality, we describe estimation of $\tau(z_0)$ and its derivative, since estimation of $\delta(z_0)$ proceeds analogously with $A$ replacing $Y$. 

As shown in equation~\eqref{eq:liv-identification}, the \textit{derivatives} of the dose-response curves $\tau(z_0)$ and $\delta(z_0)$ are components of the LIV curve. In practice,  these quantities can be independently informative as well; for example, the derivative of $\tau$ provides insight into whether practitioners should increase or decrease $z$ to maximize the average outcome locally. The derivative of $\delta$ can also be interpreted as the density of the latent threshold $T$, as shown in equation~\eqref{eq:liv-identification-density}. Such quantities have also been studied in other works in the literature \citep{colangelo2020double, bong2023local,zhang2025doubly}. Notably, \cite{zhang2025doubly} recently proposed a doubly robust estimator for the derivative of the dose-response curve and extended it to settings with positivity violations. However, their approach relies on modeling the partial derivative of the
outcome model, which can be challenging when the covariates are high-dimensional. 
In Sections \ref{sec:loc_poly} and \ref{sec:smoothing}, we propose doubly robust methods for estimating the derivative of the dose-response curve that circumvent the need to model the partial derivative. Our analysis  extends to general smooth nuisance functions, ensuring greater flexibility and robustness in practical applications. Importantly, we establish a connection between dose-response derivative estimation and LIV curve estimation, highlighting the close relationship between dose-response estimation and treatment effect estimation with a continuous IV. \\

\subsubsection{Maximal Complier Class and Local Average Treatment Effects}
\label{sec:maximal-complier}

Next, we outline an estimand that is particularly relevant to IV designs with a continuous instrument. First, we assume there is a valid bounded instrumental variable $Z \in [0,1]$. Of particular interest is how many people in the study population could possible be encouraged to take the treatment by increasing the instrument from its minimum to its maximum? In formal terms, to answer this question, we are interested in what we call the maximal complier class proportion:
 \[
\MP(A^1 > A^0).
\]
Maximality of the compliance class $\{A^1 > A^0\}$, relative to $\{A^z > A^{z'}\}$ for arbitrary $z > z'$, is implied by the monotonicity assumption, and for binary instruments $\MP(A^1 > A^0)$ is referred to as the strength of the instrument under monotonicity. Of obvious interest is the treatment effect within this maximal complier class, since this is the subpopulation whose treatment status can be influenced by changes in the instrumental variable $Z$. Formally, the objective is to identify and estimate the LATE in this group:
\[
\mathbb{E}\left[Y^{a=1}-Y^{a=0} \mid A^{z=1}>A^{z=0}\right]
\]

\noindent Notably, the LATE among the maximal complier class is a special case of $\text{LATE}(z,z')$ with $z=1,z'=0$. Under Assumptions \ref{ass:consistency}--\ref{ass:mono} and Assumption~\ref{ass:instru}, it is the case that the relevance assumption $\MP(A^1 = A^0) < 1$ holds. Moreover, an identical argument to that in \cite{angrist1996identification} proves that the proportion of the maximal complier class can be identified as
$$
\label{eq:maximal-complier}
\MP(A^1 > A^0) = \ME[\ME(A \mid \bX, Z=1)] -\ME[\ME(A \mid \bX, Z=0)]=\ME[\lambda(\bX,1)-\lambda(\bX,0)].
$$
and the LATE can be identified as
\begin{equation}
\label{eq:eff-max-complier}
\begin{aligned}
    \mathbb{E}&\left(Y^{a=1}-Y^{a=0} \mid A^{z=1}>A^{z=0}\right)\\
    =&\, \frac{\mathbb{E}[\mathbb{E}(Y \mid \bX, Z=1)-\mathbb{E}(Y \mid \bX, Z=0)]}{\mathbb{E}[\mathbb{E}(A \mid \bX, Z=1)-\mathbb{E}(A \mid \bX, Z=0)]} =      \frac{\mathbb{E}[\mu(\bX,1)-\mu(\bX,0)]}{\mathbb{E}[\lambda(\bX,1)-\lambda(\bX,0)]}  .
\end{aligned}
\end{equation}
This estimand is comprised of the terms $\tau(z)$ and $\delta(z)$ for $z=0,1$. Both of these terms have the same structure as the dose-response curve evaluated at the boundary \citep{kennedy2017non, schindl2024incremental}. In Section \ref{sec:dose-res-boundary}, we study the estimation of the dose-response function at the boundary to assess treatment effects among the maximal complier class. At first glance, the expression for the treatment effect among the maximal complier class in \eqref{eq:eff-max-complier} may appear more complex and difficult to estimate than the LIV curve in \eqref{eq:liv-identification}. However, as our analysis in the following sections reveals, the LIV curve is actually more challenging to estimate due to its reliance on derivatives, which leads to slower convergence rates. \\

\subsection{A Framework for Doubly Robust Dose-response Function Estimation}\label{sec:framework}

Finally, we outline the doubly robust framework we use to derive our estimators. One approach to estimation would be to use plug-in estimators. For example, the formula $\tau(z_0)= \mathbb{E}\{\mathbb{E}(Y \mid \mathbf{X}, Z=z) \}|_{z=z_0} =  \mathbb{E}[\mu(\bX,z)]|_{z=z_0}$ suggests the following plug-in estimator
\[
\hat{\theta}(z_0) = \MP_n\left[\hat{\mu}(\bX,z_0)\right],
\]
where $\hat{\mu}$ is an estimator for the outcome model $\mu$. However, plug-in-style estimators often suffer from bias due to nuisance estimation error, since the accuracy of the plug-in estimator depends on the estimation error in $\hat{\mu}$. When $\mu$ is difficult to estimate—such as when no prior knowledge of its parametric form is available or when it is non-smooth---the plug-in estimator will inherit the bias in $\hat{\mu}$, leading to suboptimal performance. Here, we say the plug-in style estimator has first-order bias, since it will inherit any bias present in the estimates of the nuisance functions such as $\hat{\mu}$. First order bias may result in not achieving optimal rates of convergence or asymptotic normality. \\

One alternative is to use influence function (IF) based estimation \citep{bickel1993efficient, kennedy2024semiparametric}. IF based estimation allows researchers to construct estimators that are doubly robust and have second-order bias. Such estimators yield fast parametric convergence rates even when nuisance functions are estimated at slower rates with machine learning methods. However, the dose-response function and its derivative considered in this work are not pathwise-differentiable \citep{diaz2013targeted, kennedy2017non}, preventing the direct application of standard IF approaches. To address this challenge, we apply efficiency theory to smoothed functionals of the dose-response function, summarized as follows. Specifically, to estimate the dose-response function 
\[
\tau(z_0) = \mathbb{E}[\mathbb{E}(Y \mid \bX,Z=z_0)], \quad z_0 \in \mathcal{Z}_0,
\]
consider the following weighted least-squares problem:
\begin{equation}\label{eq:wls}
    \min_{\boldsymbol{\beta}} \int   
K_h(z-z_0) \left(\tau(z)-\bg_h^{\top}(z-z_0) \boldsymbol{\beta}\right)^2 w(z) dz,
\end{equation}
where $K_h$ is a kernel function that puts more weight to points closer to $z_0$, $\bg_h(z-z_0)$ is the rescaled local basis, and $w$ is a weight function. 
Denoting the optimal solution as 
\begin{equation}\label{eq:wls-sol}
    \boldsymbol{\beta}_{wh}^*(z_0):=\left(\int \bg_h(z-z_0)K_h(z-z_0)\bg_h^{\top}(z-z_0)w(z) dz \right)^{-1} \int  \bg_h(z-z_0)K_h(z-z_0)\tau(z)w(z)dz,
\end{equation}
where we assume the matrix $\int \bg_h(z-z_0)K_h(z-z_0)\bg_h^{\top}(z-z_0)w(z) dz$ is invertible. We can interpret $\bg_h^{\top}(0)\boldsymbol{\beta}_{wh}^*(z_0)$ as a locally weighted projection of $\tau$ around $z_0$. Since this parameter is often pathwise differentiable, influence function-based approaches can be applied. This approximation technique has been applied in various contexts, including dose-response function estimation \citep{branson2023causal}, IV-based bounds on causal effects \citep{levis2023covariate}, and heterogeneous treatment effects estimation \citep{kennedy2024minimax}.
By combining the approximation error of $\bg_h^{\top}(0)\boldsymbol{\beta}_{wh}^*(z_0)$ with the properties of the influence function-based estimator, we can establish its estimation guarantees, including error bounds and asymptotic normality. The estimation error of these estimators depends on the product of nuisance estimation rates, making them more robust to nuisance estimation errors. All our estimators are derived within this framework, and we specify the particular choices of $\bg,w$ when discussing each estimator in the following sections.

\section{Dose-response Estimation at the Boundary}
\label{sec:dose-res-boundary}

In Section \ref{sec:maximal-complier}, we demonstrated how to reduce the problem of estimating the 
local treatment effect among the maximal complier class (and the maximal complier class proportion) to two separate dose-response estimation problems on the boundary of their supports. There are many existing methods for estimating the dose-response functions in the literature \citep{diaz2013targeted,semenova2021debiased,kennedy2017non, branson2023causal}. Notably, \cite{kennedy2017non} proposed a regression-based estimator for the dose-response function. Specifically, to estimate the function $\tau(z_0)=\ME[\mu(\bX,z_0)]$, we construct the following pseudo-outcome:
\[
\xi(\bO; \bar{\pi}, \bar{\mu}) := \frac{Y-\bar{\mu}(\mathbf{X}, Z)}{\bar{\pi}(Z \mid \mathbf{X})} \int_{\mathcal{X}} \bar{\pi}(Z \mid \mathbf{x}) d \MP(\mathbf{x})+\int_{\mathcal{X}} \bar{\mu}(\mathbf{x}, Z) d \MP(\mathbf{x}),
\]
where $\bar{\pi},\bar{\mu}$ are functions that may differ from the true propensity score $\pi$ and regression function $\mu$.
\cite{kennedy2017non} showed that
\[
\ME[\xi(\bO; \bar{\pi}, \bar{\mu})\mid Z=z]|_{z=z_0} = \tau(z_0)
\]
if either $\bar{\pi}=\pi$ or $\bar{\mu}=\mu$. Hence as long as either the propensity score 
or the outcome model is correctly specified, regressing $\xi(\bO; \bar{\pi}, \bar{\mu})$ on $Z$ yields the dose-response function $\tau$. This motivates Algorithm \ref{alg:DR-dose-response} in the Appendix for estimating the dose-response function \citep{bonvini2022fast} and its derivative via local polynomial regression, which will be useful in the next section. \\

We now demonstrate how equation \eqref{eq:wls} connects to Algorithm \ref{alg:DR-dose-response}. Let $\bg_h$ be the local polynomial basis and $w$ be the marginal density of $Z$. Then, the solution \eqref{eq:wls-sol} simplifies to
\[
\boldsymbol{\beta}_{wh}^*(z_0):=\left(\mathbb{E}[ \bg_h(Z-z_0)K_h(Z-z_0)\bg_h^{\top}(Z-z_0)] \right)^{-1} \mathbb{E}[\bg_h(Z-z_0)K_h(Z-z_0)\tau(Z)] ,
\]
which corresponds to the population version of the local polynomial coefficient estimator:
\[
\hat{\boldsymbol{\beta}}_{wh}(z_0):=\left(\mathbb{P}_n[ \bg_h(Z-z_0)K_h(Z-z_0)\bg_h^{\top}(Z-z_0)] \right)^{-1} \mathbb{P}_n[\bg_h(Z-z_0)K_h(Z-z_0)\xi(\bO)] .
\]
Here, $\xi(\bO)$ is the pseudo-outcome introduced in \cite{kennedy2017non}. We  show that $\hat{\boldsymbol{\beta}}_{wh}(z_0)$ is centered around $\boldsymbol{\beta}_{wh}^*(z_0)$. 
Thus, the local polynomial estimator of the dose-response function effectively estimates the smoothed function $\bg_h(0)^\top \boldsymbol{\beta}_{wh}^*(z_0)$, which corresponds to the first component of $\boldsymbol{\beta}_{wh}^*(z_0)$. Our local polynomial estimator in Section \ref{sec:loc_poly} further extends this idea, using the second component of $\boldsymbol{\beta}_{wh}^*(z_0)$ as an approximation for the derivative of the dose-response function. \\

\cite{kennedy2017non} proved that the error contribution from nuisance function estimation is second-order (i.e., in the form of a product of the convergence rates of $\hat{\mu}$ and $\hat{\pi}$). See also \cite{bonvini2022fast} for further discussion and a high-order estimator for the dose-response curve. However, these results apply only when $z_0$ is an interior point of the support $\mathcal{Z}$. Estimating the proportion of the maximal complier class and the treatment effects within this class requires evaluating the dose-response curve at the boundary. \\

In the regression function estimation literature, most regression smoothers exhibit slower convergence rates at boundary points than at interior points, a phenomenon known as ``boundary effects" \citep{gasser1979kernel}. Near boundaries, there tend to fewer data points available leading to less stable estimates and increased variability. Various methods have been proposed to address estimation issues at boundaries \citep{fan1992variable, muller1993boundary, gasser1985kernels, ruppert1994multivariate}. Notably, the local polynomial estimator adapts naturally to boundaries by fitting a higher-degree polynomial at boundary points, eliminating the need for additional boundary modifications \citep{fan1992variable, ruppert1994multivariate}. Given that the dose-response estimation problem can be framed as a regression problem, we show that local polynomial estimators also adapt to boundaries in dose-response estimation. In the following discussion, we assume $\mathcal{Z}=[0,1]$ and focus on estimating $\tau(z_0)$ for $z_0 = ch$, where $0 \leq c < 1$ (i.e., the point $z_0$ is on the left boundary).  \\

The following theorem establishes the consistency of the local polynomial estimator $\hat{\tau}(z_0)$. \\

\begin{theorem}\label{thm:LP-boundary}
    Assume the nuisance functions, their estimates, and the outcome satisfy $\epsilon \leq \pi, \hat{\pi} \leq C, |Y|,|\mu|\leq C$. The kernel is a bounded probability density supported on $[-1,1]$ with the bandwidth satisfying $h \rightarrow 0, nh \rightarrow \infty$ as $n \rightarrow \infty$. Then for the local polynomial estimator $\hat{\tau}$ evaluated at the left boundary $z_0 = ch$ for a constant $c \in [0,1)$, we have 
    \[
    \hat{\tau}(z_0) - \tau(z_0) = \tilde{\tau}(z_0) -\tau(z_0) +R_1 +R_2 ,
    \]
    \[
        R_1=O_{\MP}\left(\frac{1}{\sqrt{n^2h}}+ \frac{1}{\sqrt{nh}} \max \left\{ r_n(z_0),s_n(z_0)\right\}\right),
    \]
    \[
    R_2 = O_{\MP} \left(\frac{1}{\sqrt{n} }+ r_n(z_0)s_n(z_0) \right),
    \]
    where $\tilde{\tau}$ is the ``oracle" estimator obtained by regressing the true pseudo-outcome $\xi$ on $Z$. As a consequence, if we assume $\tau$ is $\gamma$-smooth for $\gamma \in \mathbb{N}_+$, $ \lim_{z \rightarrow 0^+}f(z) >0$ and $f, \tau^{(\gamma)}, \sigma^2$ are right continuous at $z=0$, then for $p=\lfloor \gamma\rfloor $ we have
    \[
    \hat{\tau}(z_0) - \tau(z_0)=O_{\MP}\left (h^{\gamma} + \frac{1}{\sqrt{nh}} + r_n(z_0)s_n(z_0) \right).
    \]
\end{theorem}

\bigskip

In the error decomposition of $\hat{\tau}$, the term $h^{\gamma} + \frac{1}{\sqrt{nh}}$ represents the oracle rate for estimating a $\gamma$-smooth function, while the remaining term captures the product of convergence rates for nuisance parameter estimation. 
Theorem \ref{thm:LP-boundary} shows that the local polynomial estimator achieves the same convergence rate for estimating the dose-response function at boundary points as it does at interior points, demonstrating that it automatically adapts to boundaries in dose-response estimation problems. Our results also extend those of \citep{kennedy2017non} to the general class of smooth functions using local polynomial estimators. To illustrate the final rate, we impose the following smoothness assumptions on the nuisance parameters. \\

\begin{assumption}[Smoothness]
    Assume $\pi, \mu, \tau$ belong to H\"{o}lder smooth function class:
    \label{ass:smoothness}
    \begin{itemize}
        \item $\pi$ is $\alpha$-smooth.
        \item $\mu$ is $\beta$-smooth in $\bx$ and $\gamma$-smooth in $z$.
        \item $\tau$ is $\gamma$-smooth.
    \end{itemize}
    And $\pi, \mu$ are estimated at corresponding minimax rates in the sense that
    \[
    r_n(z_0)  \asymp  n^{-\frac{1}{2+\frac{d+1}{\alpha}}},\,s_n(z_0)   \asymp n^{-\frac{1}{2+\frac{1}{\gamma}+\frac{d}{\beta}}}.
    \]
\end{assumption}

\bigskip

\noindent Note that the smoothness of $\gamma$ matches that of $\mu$ in the direction of $z$, as the smoothness of $\gamma$ can be inferred from that of $\mu$ under mild conditions. Under the smoothness assumptions specified in Assumption \ref{ass:smoothness}, we obtain the following estimation rate for the local polynomial estimator $\hat{\tau}$. \\

\begin{theorem}\label{thm:dose-response-smooth}
    Under conditions in Theorem \ref{thm:LP-boundary} and further assume Assumption \ref{ass:smoothness}, we have
    \[
\hat{\tau}(z_0)-\tau(z_0) =
\begin{cases} 
    O_{\MP}\left(  n^{-\frac{\gamma}{2\gamma+1}} \right), & \text{if } \frac{d/\beta}{(2+1/\gamma)(2+1/\gamma+d/\beta)} \leq \frac{\alpha}{2\alpha+d+1}, \\
    O_{\MP} \left( n^{-\left( \frac{1}{2+\frac{1}{\gamma}+\frac{d}{\beta}}+ \frac{1}{2+\frac{d+1}{\alpha}} \right)} \right), & \text{if } \frac{d/\beta}{(2+1/\gamma)(2+1/\gamma+d/\beta)} > \frac{\alpha}{2\alpha+d+1}.
\end{cases}
\]
\end{theorem}

\bigskip

\noindent Theorem \ref{thm:dose-response-smooth} shows that the final rate of $\hat{\tau}$ depends on the relationship among the smoothness parameters $\alpha, \beta, \gamma$. In the oracle regime
\begin{equation}\label{eq:oracle-nuisance1}
    \frac{d/\beta}{(2+1/\gamma)(2+1/\gamma+d/\beta)} \leq \frac{\alpha}{2\alpha+d+1},
\end{equation}
the nuisance functions can be estimated at sufficiently fast rates, allowing $\hat{\tau}$ to achieve the oracle rate for estimating a univariate $\gamma$-smooth function. In the alternative regime, the nuisance estimation error dominates; therefore, $\hat{\tau}$ inherits the slow convergence rates of the nuisance estimation and cannot achieve the oracle rate.

\section{A Local Polynomial Estimator for Derivative Estimation}
\label{sec:loc_poly} 

Since the LIV curve is identified as the ratio of derivatives of dose response curves, we next develop a local polynomial-based derivative estimator. Here, the derivative is estimated by the local scope of the fitted polynomial. Mathematically, since we can express the derivative function as 
\[
\theta(z_0) = \tau'(z_0), \, \tau(z_0)=\ME[\xi(\bO)\mid Z=z_0],
\]
after solving the following ``oracle" local polynomial optimization problem:
\[
\tilde{\boldsymbol{\beta}}_h(z_0)=\underset{\boldsymbol{\beta} \in \mathbb{R}^{p+1}}{\arg \min} \, \mathbb{P}_n\left[K_{h }(Z-z_0)\left\{{\xi}(\bO )-\bg_{h }(Z-z_0)^{\mathrm{T}} \boldsymbol{\beta}\right\}^2\right],
\]
the ``oracle" estimator for $\theta(z_0)$ is then given by $\tilde{\theta}(z_0)= \be_2^{\top} \tilde{\boldsymbol{\beta}}_h(z_0)/h$. However, this estimator is not feasible since the pseudo-outcome $\xi$ is not directly observed and needs to be estimated in the first stage. Following a similar approach to dose-response estimation, we first estimate the nuisance functions to impute the pseudo-outcome $\xi$, and then apply a local polynomial regression to estimate $\theta$, as detailed in Algorithm \ref{alg:DR-dose-response}. The following lemma characterizes the difference between $\hat{\theta}$ and its oracle counterpart $\tilde{\theta}$.

\begin{lemma}\label{lemma:LP-oracle}
    Assume the nuisance functions, their estimates, and the outcome satisfy $\epsilon \leq \pi, \hat{\pi} \leq C, |Y|,|\mu|\leq C$. The kernel is a bounded probability density supported on $[-1,1]$ with the bandwidth satisfing $h \rightarrow 0, nh \rightarrow \infty$ as $n \rightarrow \infty$. Then for an interior point $z_0 \in \mathcal{Z}$ we have
    \[
    \hat{\theta}(z_0) - \theta(z_0) = \tilde{\theta}(z_0) -\theta(z_0) +R_1 +R_2 ,
    \]
    \[
    R_1=O_{\MP}\left(\frac{1}{\sqrt{n^2h^3}}+ \frac{1}{\sqrt{nh^3}} \max \{r_n(z_0),s_n(z_0)\} \right),
    \]
    \[
    R_2 =O_{\MP}\left( \frac{1}{\sqrt{nh^2} }+ \frac{1}{h}r_n(z_0)s_n(z_0)\right).
    \]
\end{lemma}

\noindent Under the smoothness assumption in Assumption \ref{ass:smoothness}, we can obtain the following estimation rate for $\hat{\theta}$ in estimating the derivative of the dose-response function.

\begin{theorem}\label{thm:deriv-smooth}
    Under conditions in Lemma \ref{lemma:LP-oracle}, further assume Assumption \ref{ass:smoothness} and additional regularity conditions for local polynomial estimators in the proof, we have
    \[
\hat{\theta}(z_0)-\theta(z_0) =
\begin{cases} 
    O_{\MP}\left(n^{-\frac{\gamma-1}{2\gamma+1}}\right), & \text{if } \frac{d/\beta}{(2+1/\gamma)(2+1/\gamma+d/\beta)} \leq \frac{\alpha}{2\alpha+d+1}, \\
    O_{\MP}\left( n^{- \frac{\gamma-1}{\gamma}\left( \frac{1}{2+\frac{1}{\gamma}+\frac{d}{\beta}}+ \frac{1}{2+\frac{d+1}{\alpha}} \right)} \right), & \text{if } \frac{d/\beta}{(2+1/\gamma)(2+1/\gamma+d/\beta)} > \frac{\alpha}{2\alpha+d+1}.
\end{cases}
\]
\end{theorem}

Similar to Theorem \ref{thm:dose-response-smooth}, Theorem \ref{thm:deriv-smooth} demonstrates that the estimation rate of  $\hat{\theta}$ depends on the relationship among the smoothness parameters. Notably, the oracle regime for achieving the oracle rate in derivative estimation is identical to the condition in \eqref{eq:oracle-nuisance1} for $\hat{\tau}$ to achieve the oracle rate in dose-response function estimation. However, the optimal rates for estimating the derivative are slower than those for the dose-response function in both smoothness regimes, emphasizing that derivative estimation is generally a more challenging task.

\subsection{Asymptotic Normality}

Next, we characterize the asymptotic normality of $\hat{\theta}(z_0)$ with the following theorem. 

\begin{theorem}\label{thm:LP-normality}
    Let $z_0 \in \mathcal{Z}$ be an interior point of the support $\mathcal{Z}$ of $Z$ and $B(z_0)$ is a neighborhood of $z_0$. Assume the following conditions hold:
    \begin{enumerate}
        \item On $B(z_0)$, the marginal density of $Z$, $f$, is continuous and $f(z_0)>0$. The dose-response $\tau$ is $\gamma$-times continuously differentiable.
        \item $h \rightarrow 0,\, nh^{2p+1} \rightarrow \infty $ as $n \rightarrow \infty$, where $p=\lfloor \gamma \rfloor$ is the order of the local polynomial.
        \item The nuisance functions, their estimates, and the outcome satisfy $\epsilon \leq \pi, \hat{\pi} \leq C, |Y|,|\mu|\leq C$.
        \item The kernel is a bounded probability density supported on $[-1,1]$. The matrix $\bS = (\mu_{i+j})_{0 \leq i, j \leq p1},  \tilde{\bS} = (\nu_{i+j})_{0 \leq i, j \leq p} \in \mathbb{R}^{(p+1)\times (p+1)}$ are non-singular, where we denote $\mu_j = \int u^j K(u) du,\, \nu_j = \int u^j K^2(u) du$.
        \item The variance function $\sigma^2(z_0) = \ME[(\varphi(\bO)-\tau(Z))^2 \mid Z=z_0]$ is continuous.
        \item The nuisance estimates satisfy
    \[
    \max \left\{ r_n(z_0),s_n(z_0)\right\} \rightarrow 0,\, \sqrt{nh} \, r_n(z_0)s_n(z_0) \rightarrow 0,
    \]
    \end{enumerate}
    then we have 
    \begin{equation}\label{eq:LP-normality}
    \sqrt{nh^3}\left(\hat{\theta}(z_0)-\theta(z_0)-\hat{B}_2 (z_0)\right) \stackrel{d}{\rightarrow} N(0, \sigma^2(z_0)V_{22}/f(z_0)).
    \end{equation}
    Here $\bV = \bS^{-1}\tilde{\bS}\bS^{-1}$ and $\hat{B}_2$ is the second component of 
    \[
    \frac{1}{(p+1)!h} \bS_n^{-1}(z_0)\ME \left[ K_h(Z-z_0) \bg_h(Z-z_0)\theta^{(p+1)}(\tilde{Z})(Z-z_0)^{p+1} \right],
    \]
    where $\bS_n(z_0)=\frac{1}{n}\sum_{i=1}^n K_h(Z_i-z_0)\bg_h(Z-z_0) \bg_h(Z-z_0)^\top$, $\tilde{Z}$ lies between $z_0$ and $Z$ satisfying 
    \[
    \theta(Z)=\sum_{j=0}^p \frac{\theta^{(j)}(z_0)(Z-z_0)^j}{j!}  + \frac{\theta^{(p+1)}(\tilde{Z})(Z-z_0)^{p+1}}{(p+1)!}.
    \]
    If we further assume $ nh^{2p+3} = O(1)$, then we have
    \begin{equation}\label{eq:LP-normality-bias}
    \sqrt{nh^3} \left(\hat{\theta}(z_0)-\theta(z_0)-B_2(z_0) \right) \stackrel{d}{\rightarrow} N(0, \sigma^2(z_0)V_{22}/f(z_0)),    
    \end{equation}
    where $B_2(z_0) $ is the second component of $\frac{1}{(p+1)!}\theta^{(p+1)}(z_0) \bS^{-1}(\mu_{p+1},\dots, \mu_{2p+1})^{\top}h^{p}$. 
    
\end{theorem}
Theorem \ref{thm:LP-normality} enables the construction of pointwise confidence intervals based on the local polynomial estimator $\hat{\theta}$. If we undersmooth and set $h \ll n^{-\frac{1}{2p+3}}$ so that the bandwidth is smaller than the optimal choice $ n^{-\frac{1}{2p+3}}=n^{-\frac{1}{2\gamma+1}}$ when nuisance estimation errors are negligible in the oracle regime \eqref{eq:oracle-nuisance1}, the confidence intervals are centered around the target derivative $\theta$. However, with the optimal choice $h \asymp n^{-\frac{1}{2p+3}}$, the confidence intervals and corresponding inference are for the smoothed function $\theta(z_0)+B_2(z_0)$ rather than $\theta(z_0)$. This is known as the bias problem \citep[Section 5.7]{wasserman2006all}, a common challenge in function estimation problems \citep{Ruppert2003semiparametric,bonvini2023flexibly}. Several approaches exist to address the bias problem, each with its own trade-offs. One approach is to estimate the second-order derivative and debias the estimator \citep{calonico2018effect, takatsu2024debiased}, but this requires additional smoothness assumptions. Another method is to undersmooth \citep{fan2022estimation}, reducing the bias asymptotically relative to the variance. However, finding a practical and reliable rule for the degree of undersmoothing remains challenging. \\

Here, we acknowledge that our inference is potentially for the smoothed function and use the asymptotic variance as an uncertainty quantification for our local polynomial estimator. Theoretically, the bias shrinks to 0 as $n \rightarrow \infty$ and the proposed estimator $\hat{\theta}(z_0)$ remains consistent for $\theta(z_0)$. Compared to estimating the dose-response function itself, the appropriate scaling for $\hat{\theta}$ is $\sqrt{nh^3}$ instead of $\sqrt{nh}$ \citep{kennedy2017non}. However, the requirement on the nuisance estimation error to be asymptotically negligible remains the same as in dose-response estimation; specifically, we require that the product of the estimation errors for $\mu$ and $\pi$ be of order $o_{\MP}(1/\sqrt{nh})$. Since we employ a doubly robust estimator, the contribution of nuisance estimation error involves a product:
\[
r_n(z_0)s_n(z_0)=\sup_{|z-z_0|\leq h} \sqrt{ \ME_{\bX} \ME_{D}(\hat{\pi}(z\mid \bX)-\pi(z \mid \bX))^2 } \sup_{|z-z_0|\leq h} \sqrt{ \ME_{\bX} \ME_{D}(\hat{\mu}(\bX,z)-\mu(\bX,z))^2},
\]
which makes it easier to meet the required nuisance estimation rate compared to a plug-in-style estimator that relies solely on $\hat{\mu}$. Therefore, flexible nonparametric machine learning methods can be used to estimate the nuisance functions, while our methods remain valid for statistical inference as long as $r_n(z_0)s_n(z_0)=o(1/\sqrt{nh})$. \\


In practice, the bandwidth can be chosen by estimating the optimal value that minimizes either the local Mean Squared Error (MSE) or the global Mean Integrated Squared Error (MISE) for derivative estimation \citep{fan2018local, herrmann2024lokern}. Additionally, we propose a data-adaptive model selection framework in Appendix \ref{sec:bw-selector}, which can also be applied to select the bandwidth for estimating the derivative of the dose-response function.

\section{A Smoothing Approach for Derivative Estimation}
\label{sec:smoothing}

In this section, we introduce an alternative approach for estimating the derivative of the dose-response curve. Similar to the smoothing approach outlined in Section \ref{sec:framework}, the key idea is to define a smooth, pathwise differentiable approximation function for $\theta$, allowing for the derivation of influence function-based estimators. Following the approach in \cite{branson2023causal}, we define an estimand that smooths across $Z$ and places greater weight on subjects near $Z=z_0$. Recall that $K$ is a symmetric kernel and $K_h(z) = K(z/h)/h$ is its rescaled version for a given bandwidth parameter $h>0$. The kernel-smoothed version of $\theta$ is defined as
\[
\theta_h(z_0) = \ME \left[ \int \frac{\partial \mu(\bX, z)}{\partial z} K_h(z-z_0) dz \right],
\]
where $\mu(\bX, z) = \ME[Y \mid \bX, Z=z]$. Assume $K$ is supported on $[-1,1]$ or satisfies $K(z) \rightarrow 0$ as $|z| \rightarrow \infty$, and applying integration by parts, we obtain
\[
\int \frac{\partial \mu(\bX, z)}{\partial z} K_h(z-z_0) dz =- \int  \mu(\bX, z)K_h'(z-z_0) dz.
\]
Thus the smooth approximation $\theta_h$ can also be expressed as
\[
\theta_h(z_0) = - \ME \left[\int  \mu(\bX, z)K_h'(z-z_0) dz \right ].
\]
Another way to motivate $\theta_h$ is by directly differentiating the smooth approximation of the dose-response function:
\[
\tau_h(z_0)= \ME \left[ \int \mu(\bX,z) K_h(z-z_0) dz \right]
\]
as defined in \cite{branson2023causal}. Note that $\tau_h$ corresponds to the solution in \eqref{eq:wls-sol} with $\bg_h = 1$ (the constant basis) and $w = 1$. Thus, this smooth approximation also falls within the general framework outlined in Section \ref{sec:framework}.
Our smooth approximation approach for $\theta$ is motivated by extending this idea to the derivative of the dose-response function. \\

As $h \rightarrow 0$, the rescaled kernel $K_h(z-z_0)$ converges to a point mass at $z_0$ and we expect $\theta_h(z_0) \rightarrow \theta(z_0)$. Since this approximation does not utilize a local polynomial basis, high-order kernels are required to accurately capture the local curvature. The following proposition formalizes these intuitions and quantifies the approximation error of $\theta_h(z_0)$ under the assumptions that $\mu$ is smooth in $z$ and $K$ is a high-order kernel.\\

\begin{prop}[Approximation Error of $\theta_h$]\label{prop:smooth-error}
    Assume $\mu(\bx,z): z \mapsto \mathbb{R}$ is $\gamma$-smooth w.r.t. $z$ for $\bx \in \mathcal{X}$ almost surely and the kernel $K$ is a $(\ell-1)$-th order kernel for $\ell = \lfloor \gamma \rfloor$ satisfying
    \[
    \begin{aligned}
         \int  K(u) d u = 1, &\,\int u^{j} K(u) d u = 0, \, 1\leq j \leq \ell-1,\\
        &\,\int |u|^{\gamma-1} |K(u)| d u < \infty.
    \end{aligned}
    \]
    Then we have the following bound on the approximation error of $\theta_h$
    \[
    |{\theta}_h(z_0)-\theta(z_0)| \leq C_1 h^{\gamma-1},
    \]
    where $C_1 = \frac{L \int |u|^{\gamma-1} |K(u)| du}{(\ell-1)!}$ and $L$ is the constant of Hölder continuity. 
\end{prop}

Proposition \ref{prop:smooth-error} demonstrates that the smoothing bias vanishes as $h \rightarrow 0$, with the rate of convergence depending on the smoothness of $\mu$. When $h$ is sufficiently small, any estimator for $\theta_h$ effectively serves as an estimator for $\theta$. Therefore, we focus on developing an estimator for $\theta_h$. By smoothing the parameter, the resulting function becomes pathwise differentiable and incorporates an influence function. Following a similar derivation to \cite{branson2023causal}, one can derive the efficient influence function of $\theta_h(z_0)$ as
\[
\varphi_h(\bO;z_0) = -K_h'(Z-z_0)  \frac{Y-\mu(\bX,Z)}{\pi(Z\mid \bX)} - \int  \mu(\bX, z)K_h'(z-z_0) dz.
\]
Let $\hat{\varphi}_h$ denote the estimated influence function, with $\mu, \pi$ replaced by $\hat{\mu}, \hat{\pi}$, respectively. The doubly robust estimator of $\theta_h(z_0)$ is then given by
\[
\hat{\theta}_h(z_0) = \MP_n [\hat{\varphi}_h(\bO; z_0)]
\]
The following proposition summarizes the bias and variance of $\hat{\theta}_h(z_0)$, conditioned on the data $D$ used to train the nuisance functions $\pi$ and $\mu$.

\begin{prop}[Bounds on the Conditional Bias and Variance]\label{prop:bound-variance} Assume $|Y|, |\hat{\mu}|, \pi(z\mid \bX) \leq C$ and $\hat{\pi}(Z\mid \bX) \geq \epsilon$ for some constant $\epsilon, C>0$. Further assume the kernel $K$ satisfies $\int |K'(u)|du, \int \left(K'(u)\right)^2 du < \infty$. Then the conditional bias of $\hat{\theta}_h(z_0)$ is bounded as 
    \[
    |\MP[\hat{\theta}_h(z_0) - \theta_h(z_0)]|  \lesssim \int |K_h'(z-z_0)| \|\hat{\mu}(\cdot, z) - \mu(\cdot, z)\|_2 \|\hat{\pi}(z\mid \cdot) - \pi(z \mid \cdot)\|_2 dz = O_{\MP} \left( \frac{1}{h} r_n(z_0)s_n(z_0)\right).
    \]
    The conditional variance of $\hat{\theta}_h(z_0)$ is bounded as 
    \[
    \operatorname{Var}\left(\hat{\theta}_h(z_0)\right) \lesssim \frac{1}{nh^3}.
    \]
    \end{prop}


Compared to the results in \cite{branson2023causal}, we explicitly characterize the dependency of the bias and variance on $h$, offering valuable insights into bandwidth selection to minimize the estimation error. Under the assumptions of Propositions \ref{prop:smooth-error}--\ref{prop:bound-variance}, and combining the approximation error, conditional bias, and variance, the estimation error of $\hat{\theta}_h(z_0)$ can be expressed as
\[
\hat{\theta}_h(z_0)-\theta(z_0)=O_{\MP} \left( h^{\gamma-1} + \frac{1}{\sqrt{nh^3}} +\frac{1}{h}r_n(z_0)s_n(z_0) \right).
\]
Under Assumption \ref{ass:smoothness}, we obtain the same error decomposition as that for the local polynomial estimator (see equation \eqref{eq:LP-deriv-error} in the Appendix \ref{appendix:proof-deriv-smooth}). Consequently, similar rate analysis there can be applied to obtain the same estimation rate in Theorem \ref{thm:deriv-smooth} for $\hat{\theta}_h(z_0)$. 

\subsection{Asymptotic Normality}
In this section, we study the asymptotic normality of $\hat{\theta}_h(z_0)$. To begin, we note the following decomposition of the error:
\[
\begin{aligned}
    \hat{\theta}_h(z_0) - \theta(z_0) = &\,\hat{\theta}_h(z_0) - {\theta}_h(z_0) + {\theta}_h(z_0)-\theta(z_0)\\
    = &\, (\MP_n-\MP) [\varphi_h(\bO;z_0)] + (\MP_n-\MP)[\hat{\varphi}_h(\bO; z_0)-\varphi_h(\bO;z_0)] \\
    &\,+ \MP[\hat{\varphi}_h(\bO; z_0)-\varphi_h(\bO;z_0)]+ {\theta}_h(z_0)-\theta(z_0).
\end{aligned}
\]
The first term, $(\MP_n-\MP) [\varphi_h(\bO;z_0)]$, is a sample average that, under appropriate scaling, converges in distribution to a Gaussian random variable asymptotically. The second term, $(\MP_n-\MP)[\hat{\varphi}_h(\bO; z_0)-\varphi_h(\bO;z_0)]$, is an empirical process term that can be bounded using sample splitting or by imposing additional complexity assumptions on the nuisance model class. The third term, $\MP[\hat{\varphi}_h(\bO; z_0)-\varphi_h(\bO;z_0)]$, is the conditional bias and can be bounded by the product of the nuisance estimation rates, as summarized in Proposition \ref{prop:bound-variance}. Finally, the last term captures the approximation error of $\theta_h$, which is bounded in Proposition \ref{prop:smooth-error}. Combining these arguments, we establish the following result on the asymptotic normality of $\hat{\theta}_h(z_0)$.

\begin{theorem}\label{thm:smooth-normality}
    Assume we estimate nuisance functions $\pi, \mu$ from a separate independent sample, and the nuisance estimates satisfy $\epsilon \leq \pi, \hat{\pi} \leq C, |Y|,|\mu|\leq C$. Further assume $\mu$ is $\gamma$-smooth w.r.t. $z$ and the kernel $K$ is a $(\ell-1)$-th order kernel for $\ell = \lfloor \gamma \rfloor$ satisfying $\int |K'(u)|du,  \int (K'(u))^2du < \infty$. Then for an interior point $z_0$ we have
    \[
    \hat{\theta}_h(z_0) - \theta_h(z_0) = (\MP_n-\MP) [\varphi_h(\bO;z_0)]+ O_{\MP} \left( \frac{1}{\sqrt{nh^3} }\max \{r_n(z_0),s_n(z_0)\}+\frac{1}{h} r_n(z_0)s_n(z_0) \right)
    \]
    As a consequence, if we further assume $\operatorname{Var}(Y \mid \bX, Z) \geq c >0,$ and as $n\rightarrow \infty$, $h\rightarrow 0, nh^{3}\rightarrow \infty$,
    \[
    \max \left\{ r_n(z_0),s_n(z_0)\right\} \rightarrow 0,\, \sqrt{nh} \, r_n(z_0)s_n(z_0) \rightarrow 0,
    \]
    then we have
    \[
    \frac{\sqrt{n}(\hat{\theta}_h(z_0) - \theta_h(z_0))}{\sigma_n} \stackrel{d}{\rightarrow} N(0,1),
    \]
    where $\sigma_n^2 = \operatorname{Var}(\varphi_h(\bO;z_0)) \asymp 1/h^3$.
    
\end{theorem}

\bigskip

We note that $\hat{\theta}_h(z_0)$ centers around the smooth approximation $\theta_h(z_0)$ in Theorem \ref{thm:smooth-normality}. By Proposition \ref{prop:smooth-error}, the smoothing error is $O(h^{\gamma -1})$. If we undersmooth and assume $n h^{2\gamma +1} \rightarrow 0$, it follows that $\sqrt{nh^3}(\theta_h(z_0) - \theta(z_0))$ becomes asymptotically negligible, allowing $\hat{\theta}_h(z_0)$ to center around $\theta(z_0)$.  
In this paper, we obtain uncertainty quantification for the smooth approximation estimator $\hat{\theta}_h(z_0)$, acknowledging that inference is effectively conducted for $\theta_h(z_0)$ as discussed in Section \ref{sec:loc_poly}. Therefore, we do not pursue undersmoothing or bias correction. Discussions on double robustness and bandwidth selection follow similarly to those in Section \ref{sec:loc_poly}.

\subsection{A Comparison of Estimation Approaches}

The local polynomial estimator of  Section \ref{sec:loc_poly} achieves the same rate as the smoothing approach; both methods use kernel smoothing and are doubly robust, but they differ in their approximation strategies, the local polynomial estimator captures local curvature using polynomials, while the smoothing approach uses a local constant basis (i.e., $\bg_h=1$) and approximates  with high-order kernels. For  the weight $w$, the smoothing approach discussed  sets $w=1$, which may place additional weight on values of $z$ with a small density, where fewer observations are available. In contrast, the local polynomial estimator uses the marginal density as the weight function, which avoids  assigning  weights according to the underlying distribution of $Z$. Thus the marginal density is often preferred. The local polynomial estimator in Section \ref{sec:loc_poly} estimates the derivative of regression functions for constructed pseudo-outcomes and its idea generalizes to broader derivative estimation methods, including splines \citep{zhou2000derivative} and empirical derivatives \citep{de2013derivative}. 
Due to the non-pathwise-differentiability of the dose-response function, various smooth approximation approaches have been proposed \citep{kennedy2017non, branson2023causal}. The framework in Section \ref{sec:framework} unifies these methods and offers potential directions for future research. For instance, one could explore alternative basis functions $\bg$ to approximate the dose-response function under different structural assumptions or find weight function $w$ that improves the asymptotic variance. 



\section{Simulation Study}
\label{sec:sims}

In this section, we use simulations to compare  with the projection approach in \cite{kennedy2019robust}. For the latter, if the working parametric model is misspecified, the estimated LIV curve represents the best approximation within the specified model class to the true LIV curve. Model misspecification can still introduce  bias, leading to large estimation error. Here, we study whether our proposed  methods  reduce bias compared to the projection approach. First, we describe the data-generating process (DGP) we use for the simulations. \\

The covariates $\bX$ are drawn from the following multivariate Gaussian distribution: $\bX = (X_1,X_2,X_3,X_4) \sim N(0, \bI_4)$.
Next, the instrument $Z$ is drawn from $N(\eta(\bX),1)$ with $\eta(\bX) = 2 + 0.1X_1 + 0.1X_2 - 0.1X_3 + 0.2X_4.$
The treatment $A$ consists of draws from $A \mid \bX, Z \sim  1+(0.1, -0.2, 0.3, 0.1)\bX + 0.1Z+ \epsilon$ with $\epsilon \sim N(0,1)$.
Finally, $Y$ consists of draws from $Y \mid \bX, Z \sim  1+(0.2, 0.2, 0.3, -0.1)\bX + Z(-0.1X_1+0.1X_3 - 0.13^2Z^2) + \epsilon, \quad \epsilon \sim N(0,1).$ In this DGP, the derivative of the dose-response functions for the treatment and outcome are given by:
\begin{equation*}
    \theta^A(z) = \ME\bigg\{\frac{\partial\lambda(\bX,z)}{\partial z}\bigg\} =  0.1 \ \ \text{ and } \ \ 
    \theta^Y(z) = \ME\bigg\{\frac{\partial\mu(\bX,z)}{\partial z}\bigg\} =  - 3\cdot0.13^{2}z^2.   
\end{equation*}
The LIV curve is given by $\gamma(z)= -0.507z^2.$ \\

To evaluate the performance of the estimators under different nuisance estimation rates, we manually control the estimation error, which is use for simulation based evaluations \citep{zeng2023efficient, branson2023causal}. Specifically, we define the nuisance estimators as:
\begin{gather*}
    \hat{\eta}(\bX) = 2 + 0.1X_1 + 0.1X_2 - 0.1X_3 + 0.2X_4 +N(n^{-\alpha},n^{-2\alpha}), \\
    \hat{\lambda}(\bX)=1+(0.1, -0.2, 0.3, 0.1)\bX + 0.1Z+N(n^{-\alpha},n^{-2\alpha}),\\
    \hat\mu(\bX, Z) = 1+(0.2, 0.2, 0.3, -0.1)\bX + Z[-0.1X_1+0.1X_3 - 0.13^2(1+N(n^{-\alpha},n^{-2\alpha}))Z^2],
\end{gather*}
such that the estimation errors of $\hat\pi$ and $\hat\mu$ are $O_\MP(n^{-\alpha})$, allowing us to control their convergence rates through $\alpha$. We implement the local polynomial estimator proposed in Section \ref{sec:loc_poly} and the smooth approximation approach from Section \ref{sec:smoothing}, and compare their performance with the projection approach, where the working model is specified as linear: $\gamma^L(z) = \psi z.$ The projection approach is misspecified with respect to the working model. We evaluate the performance of each method using the root mean squared error (RMSE) over $S$ replications, averaged across values of $Z$, as follows:
\begin{equation*}
    \textnormal{RMSE} = \int \left[\frac{1}{S}\sum_{s=1}^S \{\hat\theta^s(z)-\theta(z)\}^2 \right]^{1/2} d \MP^*(z),
\end{equation*}
where replications $S$ is set to $100$ and $\MP^*$ is the truncated marginal distribution of $Z$. This  has been used in a number of previous simulations  \citep{kennedy2017non, branson2023causal, wu2024matching}.\\

\begin{figure}[h]
	\centering
	  \begin{subfigure}[t]{0.45\textwidth}
			\includegraphics[width=\textwidth]{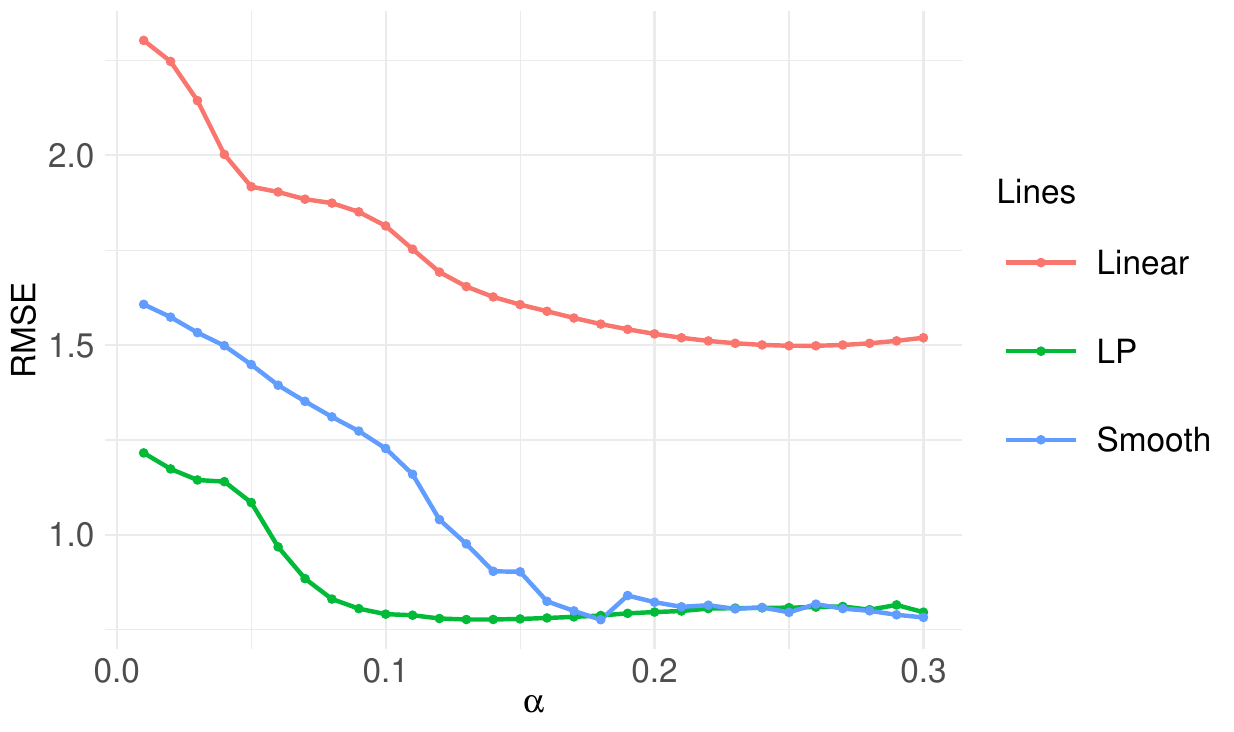}
    \caption{n=2000}
    \end{subfigure}
  \begin{subfigure}[t]{0.45\textwidth}
			\includegraphics[width=\textwidth]{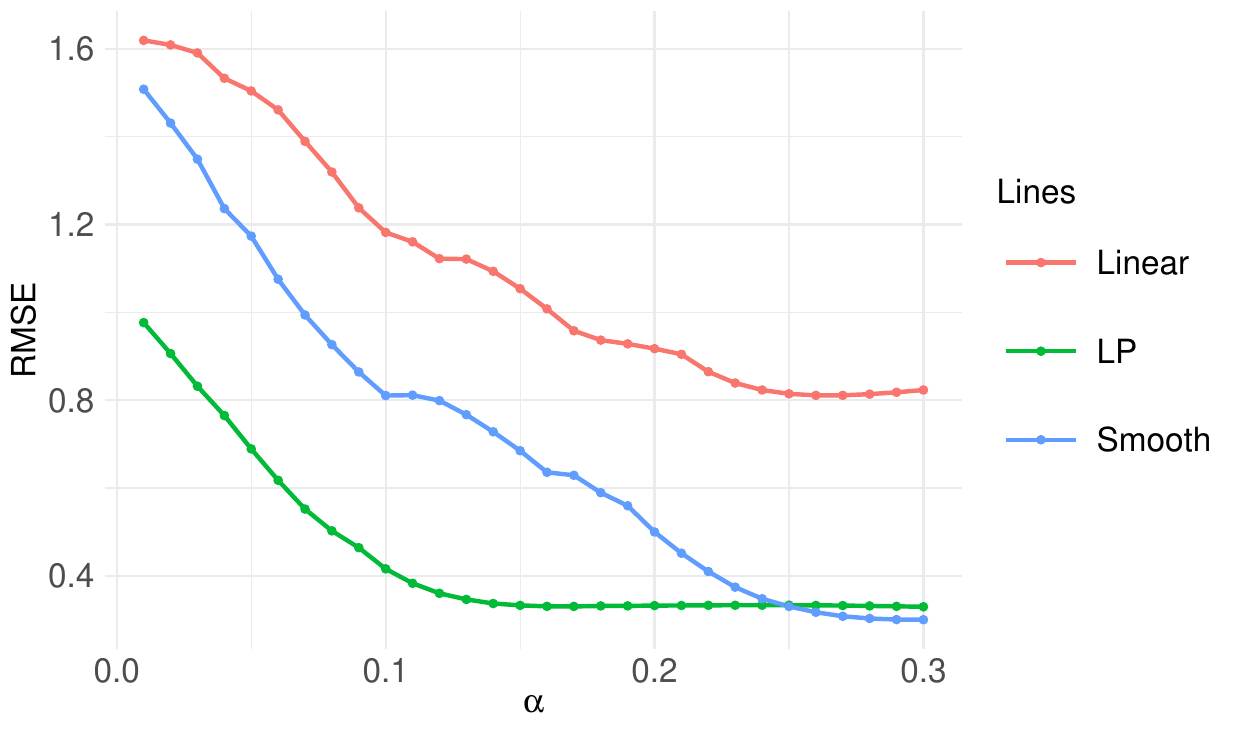}
    \caption{n=20000}
\end{subfigure}
\caption{Estimated RMSE versus $\alpha$, where $n^{-\alpha}$ is the nuisance convergence rate. }
\label{fig:rmse-misspecification}
\end{figure} 
Results are summarized in Figure \ref{fig:rmse-misspecification}. The projection approach exhibits a larger RMSE compared to  nonparametric methods. In contrast, our nonparametric methods achieve lower error and do not require prior knowledge of a correct or meaningful parametric model, making them more robust for real-world applications. Additional simulation studies that compare our doubly robust estimators with a plug-in estimator are provided in the Appendix \ref{appendix:simulation}.

\section{Application}
\label{sec:app}

In this section, we apply our proposed methodology to a study on infant mortality. The original study by \citet{Lorch:2012} aimed to estimate the effect of delivery at high-level neonatal intensive care units (NICUs) on infant mortality. High-level NICUs offer specialized delivery teams, advanced imaging capabilities, and sustained mechanical ventilation, whereas low-level NICUs are designed for routine deliveries and provide only basic care for lower-risk infants. Estimating the causal effect of high-level NICUs is challenging because they typically serve higher-risk patients, leading to potential confounding.\\

\citet{Lorch:2012} analyzed data on all premature births in Pennsylvania from 1995 to 2006. Although the dataset included baseline covariates such as birth weight, gestational age, race, and maternal comorbidities, it lacked important confounders, such as detailed physiological information about the mother and infant. Therefore, causal methods assuming no unmeasured confounding may be unreliable in this context.\\

To address this concern, \citet{Lorch:2012} used excess travel time as an instrumental variable (IV) for whether a baby was delivered at a high- versus low-level NICU. Specifically, they measured the additional travel time to the nearest high-level NICU relative to the closest low-level NICU. A greater excess travel time implies a higher cost (in time) to reach a high-level NICU, thereby discouraging some mothers from delivering there. Both \citet{Lorch:2012} and \citet{Baiocchi:2010} argue that excess travel time is a plausible instrument, as it influences the delivery location but likely has no direct effect on infant mortality.\\

We re-analyze a dataset containing information on $n=192,078$ births. In this analysis, the treatment is defined as delivery at a low-level NICU, and the instrumental variable is the measure of excess travel time. The outcome is a binary indicator of fetal death. Our goal is to estimate the proportion of deaths that could be prevented by delivery at a high-level NICU. A complete list of baseline covariates is provided in the Supplement. We begin by estimating the size of the maximal complier class and the corresponding treatment effect within this subgroup. We then estimate and compare the local instrumental variable (LIV) curves using both the local polynomial estimator and the smoothing approximation method.\\

We first estimate the size of the maximal complier class and the corresponding treatment effect within this subgroup using local polynomial estimators of the dose-response function evaluated at the boundary, as described in Section \ref{sec:dose-res-boundary}. The outcome model is estimated using an ensemble learner implemented via the \texttt{SuperLearner} package in \texttt{R}, incorporating fits from \texttt{glm}, \texttt{gam}, \texttt{ranger}, and \texttt{glmnet}. To estimate the conditional density $\pi$, we first estimate the conditional mean and variance of $A \mid \bX, Z$ using the same ensemble learner. We then apply kernel density estimation to the standardized residuals, defined as $(A - \mathbb{E}[A \mid \bX, Z]) / \operatorname{Var}(A \mid \bX, Z)^{1/2}$. The estimated proportion of the maximal complier class is 85\%, with a 95\% confidence interval of 78\% to 92\%, suggesting that a large share of mothers could be influenced to deliver at a low-level NICU (or discouraged from delivering at a high-level NICU) by varying excess travel time from its minimum to maximum. The estimated treatment effect within this complier class is 15 fewer deaths per 1,000 births, with a 95\% confidence interval of 7.7 to 22.5. This estimate is larger than the one reported in \citet{kennedy2019robust}, which relied on a parametric working model, though the confidence intervals overlap. Our results are also comparable in magnitude to those in \citet{Lorch:2012} and \citet{Baiocchi:2010}.\\

\begin{figure}[ht]
  \centering
  \begin{subfigure}[t]{0.45\textwidth}
    \includegraphics[width=\textwidth]{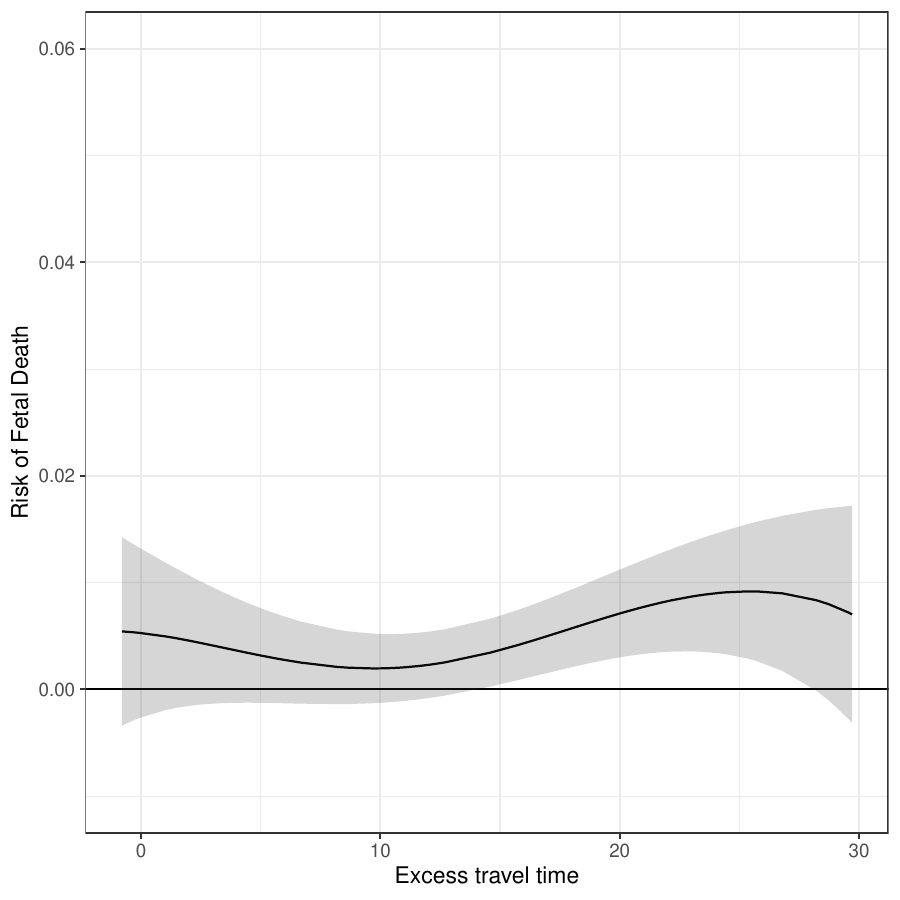}
    \caption{Local Polynomial}
   \label{fig:lp}
  \end{subfigure}
  \begin{subfigure}[t]{0.45\textwidth}
    \includegraphics[width=\textwidth]{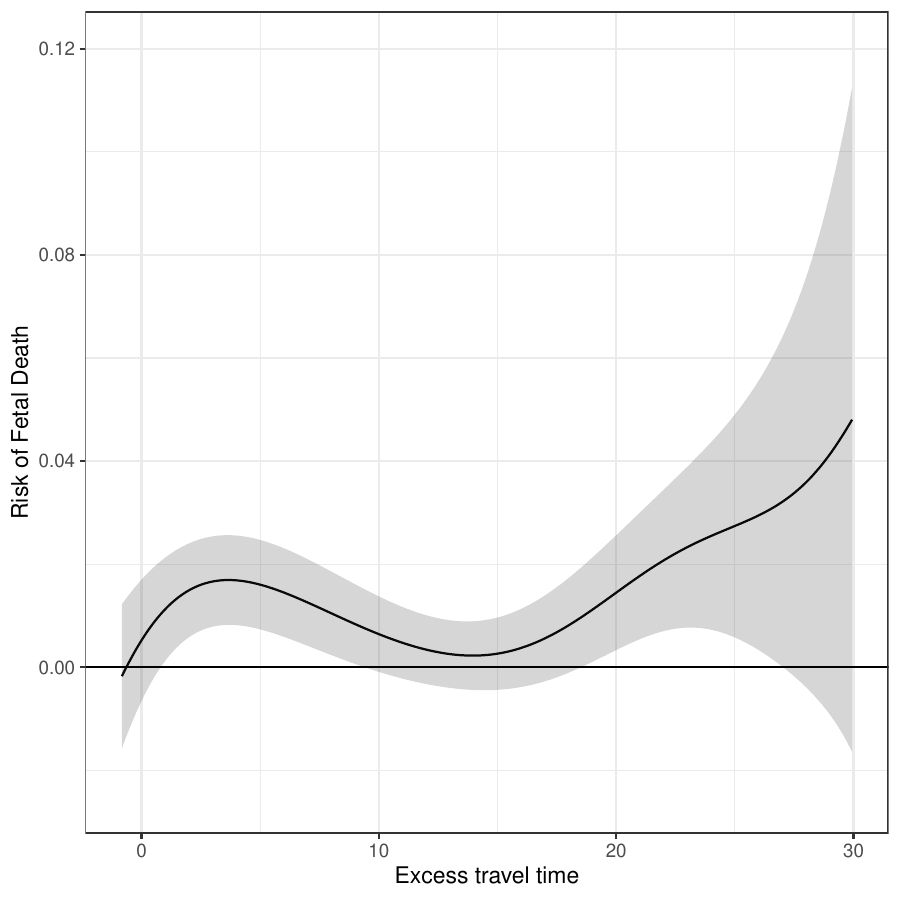}
    \caption{Smooth Approximation}
  \end{subfigure}
	\caption{LIV curve estimates via the local polynomial and smooth approximation approaches.}
	\label{fig:liv}
\end{figure} 

Estimation of the LIV curve involves two key components. First, we use the same ensemble learning approach to estimate the nuisance functions. Second, we apply the data-adaptive bandwidth selection procedures described in Appendix \ref{sec:bw-selector} for both the local polynomial estimator and the smoothing approximation method. In our analysis, varying the bandwidth within a reasonable range produced similar results, suggesting that the findings are robust to the choice of bandwidth. We restrict the LIV curve estimation to IV values below 30, as we observe instability in the estimates beyond this range due to data sparsity. This suggests that few individuals in the population have latent thresholds $T$ exceeding 30. \\

Figure~\ref{fig:liv} presents the LIV estimates. Overall, both methods exhibit similar patterns, indicating that delivery at a low-level NICU increases the risk of fetal death. Notably, the confidence intervals from both approaches nearly align along a horizontal line, suggesting that a constant-effect model may adequately capture the LIV curve. This finding is consistent with the choice of a constant-effect model in \citet{kennedy2019robust}, which relied on parametric working models for estimation.


\section{Discussion}
\label{sec:discussion}

In this paper, we study the problem of nonparametric estimation of treatment effects with a continuous IV. We begin by introducing two key estimands in the continuous IV setting: the LIV curve and the LATE. We then focus on estimating the dose-response function $\tau$ at the boundary and its derivative. To estimate $\tau$ at the boundary, we extend the approach in \cite{kennedy2017non} to a local polynomial estimator and generalize their analysis to accommodate nuisance functions with arbitrary smoothness. We further propose two doubly robust methods for estimating the derivative of the dose-response function and establish their theoretical properties (e.g., estimation rates and asymptotic normality) under appropriate conditions. The LIV curve is then obtained by taking the ratio of two estimated derivatives. All proposed methods are fully nonparametric and doubly robust, allowing for the use of flexible machine learning techniques for nuisance estimation while ensuring valid statistical inference. Finally, we illustrate our methods with an empirical study that uses excess travel time as an instrument to evaluate the treatment effects of high-level NICU on the risk of fetal death. \\

There are several possible extensions for future work. As pointed out in Section \ref{sec:assump}, the LATE estimation problem can be reduced to estimating the dose-response function at the boundary. Recently, \cite{schindl2024incremental} proposed an incremental propensity score method for estimating the boundary value of the dose-response function. Their analysis, however, is limited to Lipschitz-continuous outcome models. It would be interesting to explore whether their method can be extended to a more general smooth function class and leverage the additional smoothness for improved estimation. \\

Moreover, we address the ratio-of-derivative structure in the LIV curve by estimating the numerator and denominator separately. An important future direction is to develop methods that directly estimate the ratio. This may be particularly useful if the treatment effect (i.e., the ratio) is smoother and easier to estimate than the individual derivatives of the dose-response function, which may exhibit less smoothness and be more challenging to estimate accurately. Finally, we note a specific connection between the LIV curve and the LATE$(z,z')$:
$\gamma(z_0) = \lim_{h \rightarrow 0^+} \text{LATE}(z_0+h,z_0).$ Given the identification results in Equation~\eqref{eq:eff-max-complier} for the LATE$(z,z')$, this equation follows from the definition of the derivative. Thus, the LIV curve can be interpreted as the limit of LATE$(z,z')$ at the boundary of its domain $\{(z,z') : z > z', \, z,z' \in \mathcal{Z}\}$. Exploring estimation methods that explicitly consider this relationship between the LIV curve and LATE could provide new insights and improved techniques for estimation. \\

\bibliographystyle{apalike}
\bibliography{refer}

\newpage

\appendix

\noindent \textbf{\Large{Appendix to ``Nonparametric Estimation of Treatment Effects Using a Continuous Instrumental Variable''}}

\section{Comparison with the Marginal Treatment Effect (MTE) Framework}\label{appendix:mte}

In this section, we compare the local IV curve estimand in \citet{kennedy2019robust} with the classic marginal treatment effect (MTE) framework~\citep{heckman1999local, heckman2001policy, heckman2005structural}.

The MTE framework is built on a latent threshold selection model of the form:
\[
A = I (U \leq \lambda(\bX, Z)), \quad U \sim \operatorname{Unif}(0,1) \mid \bX,
\]
where \( A \) denotes the treatment, \( U \) is an unobserved latent variable representing individual-level resistance to treatment, and \( Z \) is a continuous instrument. The unconfounded IV assumption \( U \perp Z \mid \bX \) is often imposed. The uniform distribution of \( U \mid \bX \) serves as a normalization and does not impose a substantive restriction, since any continuous \( U \mid \bX \) can be transformed to a uniform random variable via the probability integral transform.

Under this framework, the function \( \lambda(\bX, Z) := \mathbb{P}(A = 1 \mid \bX, Z) \) is the instrument-induced treatment propensity score, which matches the notation used in our work. The variable \( U \) can then be interpreted as a resistance threshold: an individual receives treatment if and only if their latent resistance \( U \) is below the propensity score \( \lambda(\bX, Z) \).
The MTE is then defined as 
\[
\operatorname{MTE}(\bx, u):=\mathbb{E}\left[Y^1-Y^0 \mid \bX=\bx, U=u\right]
\]
and interpreted as the treatment effects among subgroups with covariates $\bX=\bx$ and resistance level $U=u$. Multiple popular causal estimands can be expressed as functionals of the MTE, including the average treatment effect (ATE), the local average treatment effect (LATE), and the policy-relevant treatment effect (PRTE). Under standard instrumental variable assumptions, the MTE is identified as
\[
\operatorname{MTE}(\bx, u) = \frac{\partial}{\partial u} \, \mathbb{E}[Y \mid \bX = \bx, \lambda(\bX, Z) = u].
\]
Since this identification involves conditioning on an unknown function of the instrument, estimation of the MTE curve is often challenging in practice. As a result, parametric models and simple plug-in estimators are commonly used.

The local IV curve in \citet{kennedy2019robust}—a conditional version of our estimand in~\eqref{eq:LIV}—is defined as
\[
\operatorname{LIV}(\bx, t) := \mathbb{E}\left[Y^1 - Y^0 \mid \bX = \bx, T = t\right],
\]
where \( T \) is a latent threshold such that \( A = I(T \leq Z) \); see the discussion following Assumption~\ref{ass:thresh} and equation~\eqref{eq:LIV}. Both the MTE and the local IV curve capture treatment effects for individuals at the margin of indifference—those who would switch treatment status in response to a small change in the instrument. The key difference lies in how this margin is modeled: the MTE is defined on the scale of the instrument-induced propensity score, while the local IV curve is defined directly on the scale of the instrument.

Importantly, in the local IV framework, the latent threshold \( T \) is allowed to have an arbitrary (continuous) distribution. In contrast, the MTE framework typically assumes a uniform latent resistance variable \( U \sim \text{Unif}(0,1) \) for identification and interpretability. Thus, the local IV curve provides a more general framework that enables identification even when the distribution of the latent threshold is unknown and non-uniform.

Moreover, estimation in the MTE framework often relies on strong parametric assumptions or plug-in estimators that assume discrete covariates \( \bX \) \citep{zeng2024causal}. These approaches may yield inconsistent estimators when the parametric model is misspecified or when covariates are high-dimensional. In contrast, \citet{kennedy2019nonparametric} established nonparametric identification of the local IV curve, and we further develop fully nonparametric, doubly robust estimators that allow for flexible machine learning methods in nuisance estimation. Our approach remains consistent under appropriate conditions even with reasonably high-dimensional covariates and misspecification of one nuisance function. See Sections~\ref{sec:loc_poly} and~\ref{sec:smoothing} for details.

\section{Detailed Estimation Algorithm for the Local Polynomial Estimator}

\begin{algorithm}[H]
\caption{Doubly Robust Estimator of the Dose-response Function and its derivative}\label{alg:DR-dose-response}
\begin{algorithmic}[1] 
    \Require Three independent samples of $n$ i.i.d observations of $\bO$ $D_1^n, D_2^n, T^n$. Here $D^n = (D_1^n, D_2^n)$ serves as the training set for estimating the nuisance functions, and pseudo-outcome regression is performed on $T^n$. 
    \Ensure Estimators of the dose-response function and its first-order derivative.
    
    
    \State Nuisance functions training: Construct estimates of $\mu, \pi$ using $D_1^n$. Then use $D_2^n$ to estimate the marginal density $f$ and get an initial estimator of $\tau(z)$ as
        \[
        \hat{f}(z) = \frac{1}{n} \sum_{i \in D_2^n} \hat{\pi}(z \mid \bX_i) , \quad\widehat{\tau}_0(z) = \frac{1}{n} \sum_{i \in D_2^n} \hat{\mu}(\bX_i, z) .
        \]
    \State Pseudo-outcome regression: Construct estimated pseudo-outcome
        \[
        \widehat{\xi}(\bO) = \frac{Y-\hat{\mu}(\bX, Z)}{\hat{\pi}(Z \mid \mathbf{X})} \hat{f}(Z)+\hat{\tau}_0(Z).
        \]
        for each observation in $T^n$ and regress the pseudo-outcomes on the treatment $Z$ in $T^n$ using local polynomial regression
        \[
        \hat{\boldsymbol{\beta}}_h(z_0)=\underset{\boldsymbol{\beta} \in \mathbb{R}^{p+1}}{\arg \min} \, \mathbb{P}_n\left[K_{h }(Z-z_0)\left\{\hat{\xi}(\bO )-\bg_{h }(Z-z_0)^{\mathrm{T}} \boldsymbol{\beta}\right\}^2\right]
        \]
        to obtain 
        \[
        \hat{\tau}(z_0) = \be_1^{\top}\hat{\boldsymbol{\beta}}_h(z_0), \quad \hat{\theta}(z_0) = \be_2^{\top} \hat{\boldsymbol{\beta}}_h(z_0)/h, z_0 \in \mathcal{Z}_0.
        \]

    \State (Optional) Cross-fitting: Swap the role of $D_1^n, D_2^n, T^n$ and repeat steps 1 and 2. Use the average of different estimates as the final estimator of $\tau(z_0), \theta(z_0)$.
    
    \noindent  \Return the estimator for dose-response $\hat{\tau}$ and its derivative $ \hat{\theta}$.
\end{algorithmic}
\end{algorithm}

\section{Adaptive Bandwidth Selection}
\label{sec:bw-selector}
In the main text, we propose two methods for estimating the derivative of the dose-response curves, both depending on a tuning parameter $h$. In this section, we propose a practical approach for model selection, which can be applied to selecting the bandwidth $h$. Specifically, let $\Theta$ be the set of candidate estimators for $\theta$. For a fixed $\bar{\theta} \in \Theta$, we evaluate its performance using the following risk function:
\[
\int (\bar{\theta}(z_0)-\theta(z_0))^2 w(z_0) d z_0,
\]
where $w$ is a weight function specified by the researcher. The model selection problem involves finding the function $\theta^{\star} \in \Theta$ that minimizes the weighted $L_2$-distance between $\bar{\theta}$ and $\theta$:
\[
\begin{aligned}
   \theta^{\star} = &\, \argmin_{\bar{\theta} \in \Theta} \int (\bar{\theta}(z_0)-\theta(z_0))^2 w(z_0) d z_0 \\
   = &\, \argmin_{\bar{\theta} \in \Theta}  \int \left(\bar{\theta}^2(z_0)-2\bar{\theta}(z_0)\theta(z_0)\right) w(z_0) d z_0.
\end{aligned}
\]
We define the pseudo-risk function as $R(\bar{\theta}) = \int (\bar{\theta}^2(z_0)-2\bar{\theta}(z_0)\theta(z_0)) w(z_0) d z_0$.
Notably, the bandwidth selection problem can be reframed as a model selection problem. Given a set of candidate bandwidths $\mathcal{H}$, the optimal bandwidth can be selected by solving the following problem:
\[
h^{\star} \in \argmin_{h \in \mathcal{H}} \int \left(\hat{\theta}_h^2(z_0) -2\hat{\theta}_h(z_0) \theta(z_0) \right)w(z_0) d z_0
\]
where $\hat{\theta}_h$ is the estimator obtained using bandwidth $h$.

In the standard cross-validation framework, the risk can typically be estimated directly from the observed outcomes. However, in our problem, the pseudo-risk depends on the unknown nuisance functions, making it challenging to estimate in a straightforward way. To address this, we derive a doubly robust loss function for  $R(\bar{\theta})$ and then apply the cross-validation framework for model selection \citep{van2003unified, kennedy2019robust}.

The key idea is to treat $R(\bar{\theta})$ as as a functional of the observed data. By deriving its influence function, we can construct a doubly robust estimator for $R(\bar{\theta})$ and hence evaluate the performance of a given candidate $\bar{\theta}$. The following proposition summarizes the influence function for $R(\bar{\theta})$. 

\begin{prop}\label{prop:IF-pseudorisk}
    Suppose the weight function $w(z)$ is continuously differentiable in $z$ and $w(z)=0$ for $z \notin \mathcal{Z}$. Further assume the candidate $\bar{\theta}$ is continuously differentiable. Under a nonparametric model, the (uncentered) influence function of $R(\bar{\theta})$ for fixed $\bar{\theta}$ and $w$ is
    \[
    L_w(\bO) = \int \bar{\theta}(z)^2 w(z )dz + 2 \left(\int \frac{d}{d z} \{ w(z) \bar{\theta}(z)\} \mu(\bX,z)dz + \left.\frac{d}{d z} \{ w(z) \bar{\theta}(z)\} \right|_{z=Z}\frac{Y-\mu(\bX,Z)}{\pi(Z\mid \bX)} \right).
    \]
\end{prop}

In practice, researchers can specify $w$ based on subject-matter considerations for learning about the curve. When such information is unavailable, a natural choice is the marginal density of $Z$, i.e., $w(z) = f(z)$. Using this choice and following the cross-validation model selection framework \citep{van2003unified}, we split the sample into two subsets, $D_1$ and $D_2$. To select a bandwidth, for each $h \in \mathcal{H}$, we use $D_1$ to obtain the nuisance functions estimates $\hat{\mu}, \hat{\pi}, \hat{f}$ and construct the estimator $\hat{\theta}_h$. The risk $R(\hat{\theta}_h)$ is then estimated on $D_2$ as:
\[
\hat{R}_2(\hat{\theta}_h) = \MP_{n_2} \left[ \hat{\theta}_h^2(Z)+2\left(\int \frac{d}{d z} \{ \hat{f}(z) \hat{\theta}_h(z)\} \hat{\mu}(\bX,z)dz + \left.\frac{d}{d z} \{ \hat{f}(z) \hat{\theta}_h(z)\} \right|_{z=Z}\frac{Y-\hat{\mu}(\bX,Z)}{\hat{\pi}(Z\mid \bX)} \right) \right],
\]
where the sample average is taken over $D_2$. To improve robustness, the roles of $D_1$ and $D_2$ can be swapped to obtain another risk estimator,$\hat{R}_2(\hat{\theta}_h)$. The bandwidth $h^{\star} $ is then selected by minimizing the combined risk estimate:
\[
\hat{R}(\hat{\theta}_h) :=(\hat{R}_1(\hat{\theta}_h)+\hat{R}_2(\hat{\theta}_h))/2.
\]
\cite{van2003unified} provides conditions under which $h^{\star} $ is asymptotically equivalent to the oracle selector that has access to the true nuisance functions. For additional details and discussion, we refer readers to \cite{van2003unified}.

When the local IV curve \eqref{eq:liv-identification} is of interest, the doubly robust cross-validation method from \cite{kennedy2019robust} can be used. This approach directly targets the local IV curve rather than separately estimating the numerator and denominator in \eqref{eq:liv-identification}, potentially leading to improved performance.

\section{Variance Estimation for the Local IV Curve}\label{appendix:variance}
In the main texts, we discuss the asymptotic distributions of the proposed local polynomial and smooth approximation estimators for the derivative of the dose-response functions. The local IV curve is the ratio of two such curves. To quantify the uncertainty of the ratio, where the numerator and the denominator can have different convergence rates, we need the following results.

\begin{lemma}\label{lemma:CLT-ratio}
    Suppose $U_n, V_n$ are sequences of random variables and $a_n, b_n, u_n, v_n$ are non-random sequences satisfying 
    \[
    a_n, b_n \rightarrow \infty,\, u_n \rightarrow \theta_U, \, v_n \rightarrow \theta_V,
    \]
    as $n\rightarrow \infty$, where $\theta_U\in \mathbb{R}, \theta_V \neq 0$. Further assume 
    \[
    a_n (U_n-u_n) \stackrel{d}{\rightarrow} N(0, \sigma_U^2), \,b_n (V_n-v_n) \stackrel{d}{\rightarrow} N(0, \sigma_V^2),
    \]
    then for the asymptotic distribution of the ratio $U_n/V_n$, we have
    \begin{enumerate}
        \item If $a_n/b_n \rightarrow \infty$, we have
        \[
        b_n \left(\frac{U_n}{V_n} -\frac{u_n}{v_n}\right)\stackrel{d}{\rightarrow} N(0, \theta_U^2\sigma_V^2/\theta_V^4).
        \]
        \item If $a_n/b_n \rightarrow 0$, we have 
        \[
        a_n \left(\frac{U_n}{V_n} -\frac{u_n}{v_n}\right)\stackrel{d}{\rightarrow} N(0, \sigma_U^2/\theta_V^2).
        \]
        \item If $a_n = b_n$ and further assume $a_n[(U_n, V_n)^\top - (u_n, v_n)^\top] \stackrel{d}{\rightarrow} N(\boldsymbol{0}, \boldsymbol{\Sigma})$, we have
        \[
        a_n \left(\frac{U_n}{V_n} -\frac{u_n}{v_n}\right)\stackrel{d}{\rightarrow} N(0, (1/\theta_V,-\theta_U/\theta_V^2) \boldsymbol{\Sigma} (1/\theta_V,-\theta_U/\theta_V^2)^\top).
        \]
    \end{enumerate}
\end{lemma}
In Lemma \ref{lemma:CLT-ratio}, the centralization terms $u_n$ and $v_n$ are allowed to depend on $n$. We can set $U_n = \hat{\theta}^Y(z_0)$ and $V_n = \hat{\theta}^A(z_0)$, with $u_n$ and $v_n$ chosen according to the estimation methods applied. This allows us to obtain the asymptotic distribution of the ratio $\hat{\theta}^Y(z_0)/\hat{\theta}^A(z_0)$ and estimate its asymptotic variance using the individual variance of $\hat{\theta}^Y(z_0),\hat{\theta}^A(z_0)$ accordingly.

The above approach may require knowledge of the convergence rates of the numerator and denominator. Alternatively, when the convergence rates are unknown and we cannot distinguish among the three cases, we can use an asymptotic expansion approach. Suppose the following asymptotic expansions hold for the numerator and denominator:
\[
\hat{\theta}_{h_1}^Y(z_0)-{\theta}_{h_1}^Y(z_0)=(\MP_n-\MP)[\phi_{h_1}^Y(\bO;z_0)] + o_{\MP}\left(1/\sqrt{nh_1^3}\right),
\]
\[
\hat{\theta}_{h_2}^A(z_0)-{\theta}_{h_2}^A(z_0)=(\MP_n-\MP)[\phi_{h_2}^A(\bO;z_0)] + o_{\MP}\left(1/\sqrt{nh_2^3}\right),
\]
where ${\theta}_{h_1}^Y(z_0), {\theta}_{h_2}^A(z_0)$ are smoothed versions of the derivative of the dose-response function as discussed in Section \ref{sec:loc_poly}--\ref{sec:smoothing}. Then by Taylor's expansion, we have
\[
\frac{\hat{\theta}_{h_1}^Y(z_0)}{\hat{\theta}_{h_2}^A(z_0)} - \frac{{\theta}_{h_1}^Y(z_0)}{{\theta}_{h_2}^A(z_0)} = (\MP_n-\MP) \left[ \frac{\phi_{h_1}^Y(\bO;z_0)}{{\theta}_{h_2}^A(z_0)}-\frac{{\theta}_{h_1}^Y(z_0)}{{\theta}_{h_2}^A(z_0)^2} \phi_{h_2}^A(\bO;z_0)\right]+o_{\MP}\left(1/\sqrt{nh_1^3}+1/\sqrt{nh_2^3}\right).
\]
The variance can then be estimated by 
\[
\frac{1}{n}\hat{\operatorname{Var}} \left(\frac{\hat{\phi}_{h_1}^Y(\bO;z_0)}{\hat{\theta}_{h_2}^A(z_0)}-\frac{\hat{\theta}_{h_1}^Y(z_0)}{\hat{\theta}_{h_2}^A(z_0)^2} \hat{\phi}_{h_2}^A(\bO;z_0)\right).
\]
It is easy to see that this approach automatically adapts to the convergence rates of the numerator and denominator without requiring prior knowledge of which one has a faster rate. The influence functions in the linear expansion for the local polynomial estimator and the smooth approximation estimator are given by (take the numerator as an example)
\[
\begin{aligned}
   \widehat{\phi}_{h }^{lp}(\mathbf{O};z_0)= &\,\frac{1}{h}\mathbf{e}_2^{\top} \widehat{\mathbf{D}}_{h z_0}^{-1} \mathbf{g}_{h}(Z-z_0) K_{h}(Z-z_0)\left(\widehat{\xi}(\mathbf{O})-\mathbf{g}_{h}^{\top}(Z-z_0) \widehat{\boldsymbol{\beta}}_h(z_0)\right)\\
   &\,+\frac{1}{h}\mathbf{e}_2^{\top} \widehat{\mathbf{D}}_{h z_0}^{-1} \int \mathbf{g}_{h}(t-z_0) K_{h }(t-z_0) \widehat{\mu}( \mathbf{X},t) d \mathbb{P}_n(t)-\widehat{\theta}_h(z_0),
\end{aligned}
\]
\[
 \widehat{\phi}_{h }^{sm}(\mathbf{O};z_0) = -K_h'(Z-z_0)  \frac{Y-\hat{\mu}(\bX,Z)}{\hat{\pi}(Z\mid \bX)} - \int  \hat{\mu}(\bX, z)K_h'(z-z_0) dz.
\]

\section{Additional Simulation Results}
\label{appendix:simulation}

In this section, we further evaluate the finite-sample properties of the proposed methods through empirical experiments. We compare the doubly robust estimators for the derivative of the dose-response function, introduced in Sections \ref{sec:loc_poly} and \ref{sec:smoothing}, with a plug-in-style estimator and illustrate their appealing properties. The data-generating process is as follows: The covariates $\bX$ are drawn from a multivariate Gaussian distribution:
\begin{equation*}
    \bX = (X_1,X_2,X_3,X_4) \sim N(0, \bI_4),  
\end{equation*}
Conditioning on the covariates $\bX$, the treatment $Z$ is sampled from $N(\eta(\bX),1)$ with
\[
\eta(\bX) = -0.8 + 0.1X_1 + 0.1X_2 - 0.1X_3 + 0.2X_4.
\]
The outcome $Y$ 
\begin{gather*}
    Y \mid \bX, Z = 1 + (0.2, 0.2, 0.3, -0.1)\bX + Z(0.1-0.1X_1+0.1X_3 - 0.13^2Z^2) + \epsilon, \quad \epsilon \sim N(0,4).
\end{gather*}
Thus, in this setup, the derivative of the dose-response function is given by:
\begin{equation*}
    \theta(z) = \ME\bigg\{\frac{\partial\mu(\bX,z)}{\partial z}\bigg\} = 0.1 - 3\cdot0.13^{2}z^2.   
\end{equation*}
To evaluate the performance of the estimators under different nuisance estimation rates, we manually control the estimation error, which is suitable for simulation purposes \citep{zeng2023efficient, branson2023causal}. Specifically, we define the nuisance estimators as:
\begin{gather*}
    \hat{\eta}(\bX) = -0.8 + 0.1X_1 + 0.1X_2 - 0.1X_3 + 0.2X_4 + N(n^{-\alpha},n^{-2\alpha}), \\
    \hat\mu(\bX, Z) = 1 + (0.2, 0.2, 0.3, -0.1)\bX + Z[0.1-0.1X_1+0.1X_3 - 0.13^2(1+N(n^{-\alpha},n^{-2\alpha}))Z^2],
\end{gather*}
such that the estimation errors of $\hat\pi$ and $\hat\mu$ are $O_\MP(n^{-\alpha})$, allowing us to control their convergence rates through $\alpha$.
We implement the local polynomial estimator proposed in Section \ref{sec:loc_poly} and the smooth approximation approach from Section \ref{sec:smoothing}, and compare their performance against the plug-in-style estimator
\[
\MP_n\big[\partial \hat\mu(\bX,z)/\partial z\big]
\]
obtained by numerical differentiation using \texttt{numDeriv} package in \texttt{R}. Following the previous simulation studies \citep{kennedy2017non, branson2023causal, wu2024matching}, we compute the root mean squared error (RMSE) over $S$ replications, averaged across a set of values of $Z$, as follows:
\begin{equation*}
    \textnormal{RMSE} = \int \left[\frac{1}{S}\sum_{s=1}^S \{\hat\theta^s(z)-\theta(z)\}^2 \right]^{1/2} d \MP^*(z),
\end{equation*}
where the number of replications $S$ is set to $100$ and $\MP^*$ is the truncated marginal distribution of $Z$. The results are summarized in Figure \ref{fig:rmse-alpha}.

\begin{figure}[h]
	\centering
	  \begin{subfigure}[t]{0.45\textwidth}
			\includegraphics[width=\textwidth]{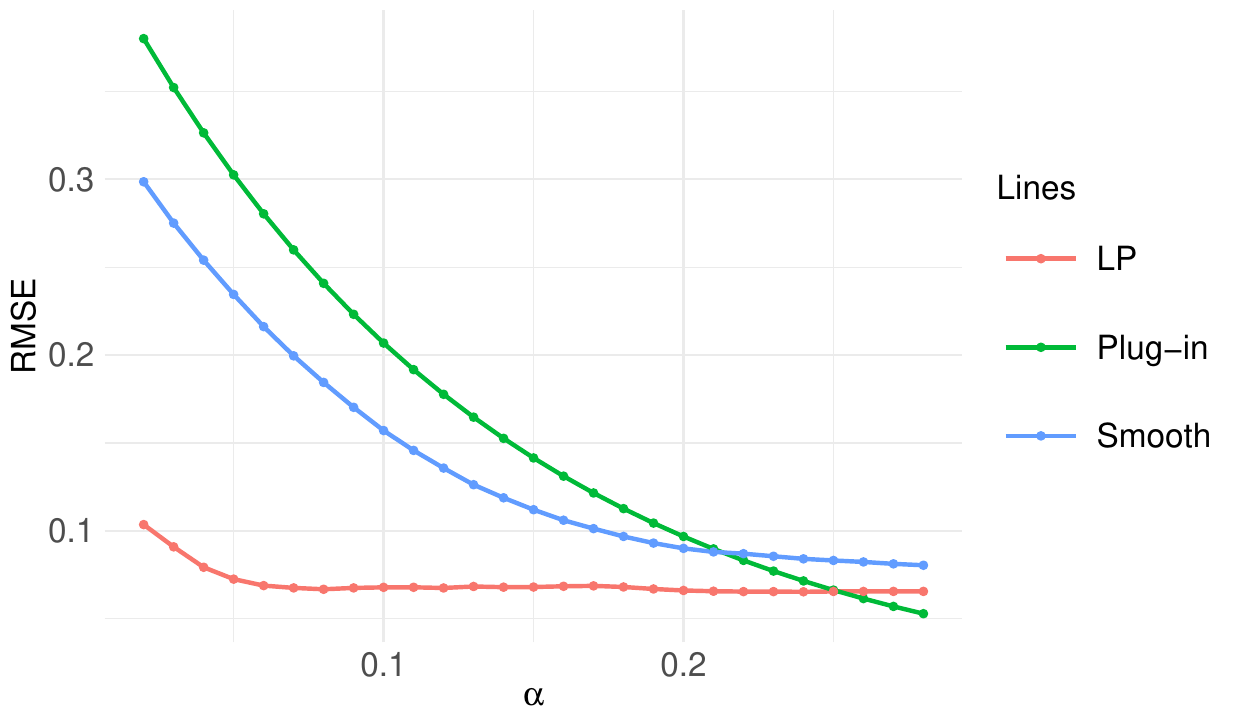}
    \caption{n=2000}
    \end{subfigure}
  \begin{subfigure}[t]{0.45\textwidth}
			\includegraphics[width=\textwidth]{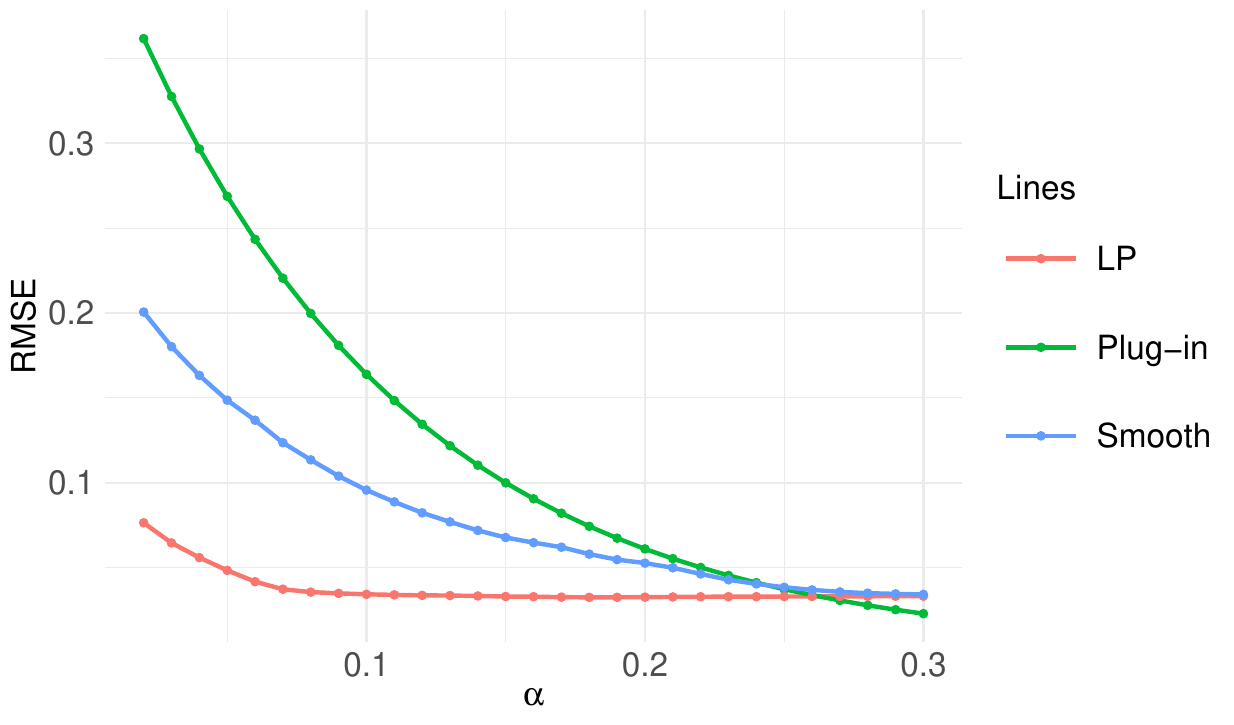}
    \caption{n=20000}
\end{subfigure}
\caption{Estimated RMSE versus $\alpha$, where $n^{-\alpha}$ is the nuisance convergence rate. }
\label{fig:rmse-alpha}
\end{figure}

As shown in Figure \ref{fig:rmse-alpha}, if the nuisance estimation error is large ($\alpha$ is small), both doubly robust estimators outperform the naive plug-in estimator. This can be attributed to the second-order bias term of the doubly robust estimators, where the conditional bias is the product of the nuisance estimation errors, making it ``doubly small." In contrast, the plug-in estimator directly inherits the slower convergence rate of $\widehat{\mu}$. However, as $\alpha$ increases and the nuisance function estimates become more accurate, the plug-in estimator eventually outperforms the doubly robust estimators. This occurs because the doubly robust estimators can suffer from accumulated errors in pseudo-outcome construction, bandwidth selection, and smoothing, which dominate the conditional bias when the nuisance estimation is sufficiently precise. 

\section{Application Details}

Here, we provide additional details on the data analysis in Section~\ref{sec:app}. Following the analysis in \citet{Lorch:2012} and \citet{Baiocchi:2010} we adjusted for 16 covariates. The first set of these covariates measures information about the zip code in which the mother lives: median income, percentage below poverty, median home value, percent with high school degree, percent with college degree, percent who rent versus own home. The second set of these covariates measures information about the mother: age, diabetes status, month prenatal care was started, number of times previously given birth, whether multiple deliveries, education level (8th grade or less, some high school, high school graduate, some college, college graduate, or more than college), mother’s race (White, Black, Asian/Pacific Islander, or other), insurance type (fee for service, HMO, federal/state, other, or uninsured). The final two covariates measured information about the infant: birthweight and gestational age.

\section{Proof of Auxiliary Lemmas}

\subsection{Proof of Lemma \ref{lemma:LP-oracle}}

\begin{proof}
    Since local polynomial estimator is linear in the response, we have
    \[
    \begin{aligned}
        \hat{\theta}(z_0) - \theta(z_0) = &\,\tilde{\theta}(z_0) - \theta(z_0) + \hat{\theta}(z_0) - \tilde{\theta}(z_0)\\
        = &\, \tilde{\theta}(z_0) - \theta(z_0) + \frac{1}{h}\be_2^{\top} \hat{\bD}_{hz_0}^{-1} \MP_n \left[\bg_h(Z-z_0)K_h(Z-z_0)\left(\hat{\xi}(\bO) - \xi(\bO)\right) \right]\\
        = &\, \tilde{\theta}(z_0) - \theta(z_0) + \frac{1}{h}\be_2^{\top} \hat{\bD}_{hz_0}^{-1} (\MP_n-\MP) \left[\bg_h(Z-z_0)K_h(Z-z_0)\left(\hat{\xi}(\bO) - \xi(\bO)\right) \right] \\
        &\,+ \frac{1}{h}\be_2^{\top} \hat{\bD}_{hz_0}^{-1} \MP \left[\bg_h(Z-z_0)K_h(Z-z_0)\left(\hat{\xi}(\bO) - \xi(\bO)\right) \right],
    \end{aligned}
    \]
    where $\hat{\bD}_{hz_0} = \MP_n [\bg_h(Z-z_0)K_h(Z-z_0)\bg_h^{\top}(Z-z_0)]$. Following the proof of \cite[Theorem 2]{kennedy2017non} we have
    \[
    \be_2^{\top} \hat{\bD}_{hz_0}^{-1} = O_{\MP}(1).
    \]
    Write 
    \[
    R_1 = \frac{1}{h}\be_2^{\top} \hat{\bD}_{hz_0}^{-1}(\MP_n-\MP) \left[\bg_h(Z-z_0)K_h(Z-z_0)\left(\hat{\xi}(\bO) - \xi(\bO)\right) \right],
    \]
    \[
    R_2 = \frac{1}{h}\be_2^{\top} \hat{\bD}_{hz_0}^{-1}\MP \left[\bg_h(Z-z_0)K_h(Z-z_0)\left(\hat{\xi}(\bO) - \xi(\bO)\right) \right].
    \]
    For $R_1$, by Lemma 2 in \cite{kennedy2020sharp} we have
    \[
    \begin{aligned}
        &\, (\MP_n-\MP) \left[g_{h,j}(Z-z_0)K_h(Z-z_0)\left(\hat{\xi}(\bO) - \xi(\bO)\right) \right]\\
        = &\, O_{\MP}\left( \frac{\left\|g_{h,j}(Z-z_0)K_h(Z-z_0)\left(\hat{\xi}(\bO) - \xi(\bO)\right)\right\|_2}{\sqrt{n}} \right)
    \end{aligned}   
    \]
    By direct calculations,
    \begin{equation}\label{eq:xi-error}
        \begin{aligned}
        \hat{\xi}(\bO) - \xi(\bO) = &\, \frac{(Y-\hat{\mu}(\bX,Z))(\hat{f}(Z)-f(Z))}{\hat{\pi}(Z \mid \bX)} + \frac{(Y-\mu(\bX,Z))(\pi(Z\mid \bX)-\hat{\pi}(Z\mid \bX))f(Z)}{\hat{\pi}(Z\mid \bX)\pi(Z\mid \bX)}\\
        &\, + \frac{(\mu(\bX,Z)-\hat{\mu}(\bX,Z))f(Z)}{\hat{\pi}(Z \mid \bX)} + \hat{\tau}_0(Z)-\tau(Z).
    \end{aligned}
    \end{equation}
    For the first term in \eqref{eq:xi-error},
    \[
    \begin{aligned}
        &\,\MP\left[g_{h,j}^2(Z-z_0)K_h^2(Z-z_0)\frac{(Y-\hat{\mu}(\bX,Z))^2(\hat{f}(Z)-f(Z))^2}{\hat{\pi}^2(Z \mid \bX)}\right]\\
        \lesssim &\, \MP\left[g_{h,j}^2(Z-z_0)K_h^2(Z-z_0)(\hat{f}(Z)-f(Z))^2\right]\\
        = &\, \int \left(\frac{z-z_0}{h}\right)^{2(j-1)} K_h^2(z-z_0)(\hat{f}(z)-f(z))^2 f(z) dz \\
        = &\, \frac{1}{h} \int u^{2(j-1)} K^2(u) \left(\hat{f}(z_0+hu)-f(z_0+hu)\right)^2 f(z_0+hu) du
    \end{aligned}
    \]
    Take expectation over the training set $D= D^n$ and apply Fubini's Theorem, we have
    \[
    \begin{aligned}
        &\, \ME_D \left[ \MP\left(g_{h,j}^2(Z-z_0)K_h^2(Z-z_0)\frac{(Y-\hat{\mu}(\bX,Z))^2(\hat{f}(Z)-f(Z))^2}{\hat{\pi}^2(Z \mid \bX)}\right)\right]\\
        \lesssim &\, \frac{1}{h} \int u^{2(j-1)} K^2(u) \ME_D \left[(\hat{f}(z_0+hu)-f(z_0+hu))^2\right] f(z_0+hu) du \\
        \leq &\, \frac{1}{h} \sup_{|z-z_0|\leq h } \ME_D \left[(\hat{f}(z)-f(z))^2\right] \int u^{2(j-1)} K^2(u) f(z_0+hu) du \\
        \lesssim &\, \frac{1}{h}\sup_{|z-z_0|\leq h} \ME_D \left[(\hat{f}(z)-f(z))^2\right]
    \end{aligned}
    \]
    Note that
    \[
    \hat{f}(z)-f(z) = \frac{1}{n} \sum_{i \in D_2^n} \hat{\pi}(z \mid \bX_i) - \MP[\hat{\pi}(z\mid \bX)] + \int \hat{\pi}(z \mid \bx)- \pi(z \mid \bx) d \MP(\bx).
    \]
    \[
    \begin{aligned}
        &\,\ME_D\left[(\hat{f}(z)-f(z))^2\right]\\
\leq  &\, 2\left \{ \ME_D\left[\left(\frac{1}{n} \sum_{i \in D_2^n} \hat{\pi}(z \mid \bX_i) - \MP[\hat{\pi}(z\mid \bX)]\right)^2\right] + \ME_D\left[\left(\int \hat{\pi}(z \mid \bx)- \pi(z \mid \bx) d \MP(\bx)\right)^2\right] \right\}
    \end{aligned}
    \]
    By Chebyshev's inequality we have
    \[
    \ME_{D_2^n}\left[\left(\frac{1}{n} \sum_{i \in D_2^n} \hat{\pi}(z \mid \bX_i) - \MP[\hat{\pi}(z\mid \bX)]\right)^2\right] \leq \frac{1}{n} \operatorname{Var}\left(\hat{\pi}(z \mid \bX) \mid D_1^n\right) \lesssim \frac{1}{n}.
    \]
    Hence
    \[
    \sup_{|z-z_0|\leq h }\ME_D\left[\left(\frac{1}{n} \sum_{i \in D_2^n} \hat{\pi}(z \mid \bX_i) - \MP[\hat{\pi}(z\mid \bX)]\right)^2\right] \lesssim \frac{1}{n}.
    \]
    For the second term, we have
    \[
    \begin{aligned}
        &\,\ME_D\left[\left(\int \hat{\pi}(z \mid \bx)- \pi(z \mid \bx) d \MP(\bx)\right)^2\right]\\
        \leq &\, \ME_D\left[\int (\hat{\pi}(z \mid \bx)- \pi(z \mid \bx))^2 d \MP(\bx)\right]\\
        = &\, \ME_{\bX} \left[ \ME_D(\hat{\pi}(z \mid \bX)- \pi(z \mid \bX))^2 \right] 
    \end{aligned}
    \]
    Thus we have
    \[
    \sup_{|z-z_0|\leq h} \ME_D \left[(\hat{f}(z)-f(z))^2\right] \lesssim \frac{1}{n} +\sup_{|z-z_0|\leq h} \ME_{\bX} \left[ \ME_D(\hat{\pi}(z \mid \bX)- \pi(z \mid \bX))^2 \right]. 
    \]
    \[
    \begin{aligned}
        &\,\left\|g_{h,j}(Z-z_0)K_h(Z-z_0)\frac{(Y-\hat{\mu}(\bX,Z))(\hat{f}(Z)-f(Z))}{\hat{\pi}(Z \mid \bX)}\right\| \\
        =&\,  O_{\MP} \left(\frac{1}{\sqrt{nh}} + \sqrt{\frac{1}{h} \sup_{|z-z_0|\leq h} \ME_{\bX} \left[ \ME_D(\hat{\pi}(z \mid \bX)- \pi(z \mid \bX))^2 \right]}\right).
    \end{aligned}
    \]
    Similarly one could show  for the last term in \eqref{eq:xi-error},
    \[
    \begin{aligned}
        &\,\left\|g_{h,j}(Z-z_0)K_h(Z-z_0) (\hat{\tau}_0(Z)- \tau(Z))\right\|\\
        =&\,  O_{\MP} \left(\frac{1}{\sqrt{nh}} + \sqrt{\frac{1}{h} \sup_{|z-z_0|\leq h} \ME_{\bX} \left[ \ME_D(\hat{\mu}(\bX,z)- \mu(\bX,z))^2 \right]}\right).
    \end{aligned}
    \]
    For the third term in \eqref{eq:xi-error}, 
    \[
    \begin{aligned}
        &\, \ME_D \left\{ \ME_{\bX,Z}\left[ g_{h,j}^2(Z-z_0)K_h^2(Z-z_0) \frac{(\mu(\bX,Z)-\hat{\mu}(\bX,Z))^2f^2(Z)}{\hat{\pi}^2(Z \mid \bX)}\right] \right\}\\
        \lesssim &\, \ME_{\bX,Z} \left\{ \ME_D\left[ g_{h,j}^2(Z-z_0)K_h^2(Z-z_0) (\mu(\bX,Z)-\hat{\mu}(\bX,Z))^2\right] \right\}\\
        = &\, \int \int g_{h,j}^2(z-z_0)K_h^2(z-z_0)\ME_D\left[(\hat{\mu}(\bx,z)-{\mu}(\bx,z))^2\right] \pi(z \mid \bx) dz d \MP(\bx) \\
        \lesssim &\, \int  g_{h,j}^2(z-z_0)K_h^2(z-z_0) \int\ME_D\left[(\hat{\mu}(\bx,z)-{\mu}(\bx,z))^2\right]   d \MP(\bx) dz \\
        \leq &\, \sup_{|z-z_0|\leq h}  \ME_\bX\left[ \ME_D (\hat{\mu}(\bX,z)-\mu(\bX,z))^2\right] \int  g_{h,j}^2(z-z_0)K_h^2(z-z_0) dz \\
        \lesssim &\, \frac{1}{h}\sup_{|z-z_0|\leq h}  \ME_\bX\left[ \ME_D (\hat{\mu}(\bX,z)-\mu(\bX,z))^2\right].
    \end{aligned}
    \]
    Thus we have 
    \[
    \begin{aligned}
        &\,\left\|g_{h,j}(Z-z_0)K_h(Z-z_0)\frac{({\mu}(\bX,Z)-\hat{\mu}(\bX,Z))f(Z)}{\hat{\pi}(Z \mid \bX)}\right\| \\
        =&\,  O_{\MP} \left( \sqrt{\frac{1}{h}\sup_{|z-z_0|\leq h}  \ME_\bX\left[ \ME_D (\hat{\mu}(\bX,z)-\mu(\bX,z))^2\right]}\right).
    \end{aligned}
    \]
    Similarly for the second term in \eqref{eq:xi-error} one could show 
    \[
    \begin{aligned}
        &\,\left\|g_{h,j}(Z-z_0)K_h(Z-z_0)\frac{(Y-\mu(\bX,Z))(\pi(Z\mid \bX)-\hat{\pi}(Z\mid \bX))f(Z)}{\hat{\pi}(Z\mid \bX)\pi(Z\mid \bX)}\right\| \\
        =&\,  O_{\MP} \left( \sqrt{\frac{1}{h}\sup_{|z-z_0|\leq h}  \ME_\bX\left[ \ME_D (\hat{\pi}(z\mid \bX)-\pi(z \mid \bX))^2\right]}\right).
    \end{aligned}
    \]
    So we conclude
    \[
    \begin{aligned}
        &\,\left\|g_{h,j}(Z-z_0)K_h(Z-z_0)\left(\hat{\xi}(\bO) - \xi(\bO)\right)\right\|_2 \\
        = &\, O_{\MP}\left(\frac{1}{\sqrt{nh}}+ \frac{1}{\sqrt{h}} \max \left\{ \sup_{|z-z_0|\leq h} \sqrt{\ME_{\bX}[\ME_D \left(\hat{\mu}(\bX,z) - \mu(\bX,z)\right)^2]}, \sup_{|z-z_0|\leq h} \sqrt{\ME_{\bX}[\ME_D \left(\hat{\pi}(z \mid \bX) - \pi(z \mid \bX)\right)^2]}\right\} \right),
    \end{aligned}
    \]
    \[
    \begin{aligned}
        &\,R_1 = O_{\MP}\left(\frac{1}{\sqrt{n^2h^3}}+ \right.\\
        &\, \left.\frac{1}{\sqrt{nh^3}} \max \left\{ \sup_{|z-z_0|\leq h} \sqrt{\ME_{\bX}[\ME_D \left(\hat{\mu}(\bX,z) - \mu(\bX,z)\right)^2]}, \sup_{|z-z_0|\leq h} \sqrt{\ME_{\bX}[\ME_D \left(\hat{\pi}(z \mid \bX) - \pi(z \mid \bX)\right)^2]}\right\} \right).
    \end{aligned}
    \]
    To bound $R_2$, note that
    \[
    \begin{aligned}
        &\,\ME\left [\hat{\xi}(\bO)-\xi(\bO) \mid D, Z=z \right]\\
        = &\, \ME\left[ \frac{\mu(\bX,Z)-\hat{\mu}(\bX,Z)}{\hat{\pi}(Z\mid \bX)}\hat{f}(Z) \mid D, Z=z \right] + \hat{\tau}_0(z)-\tau(z).
    \end{aligned}
    \]
    Rewrite 
    \[
    \begin{aligned}
        \hat{\tau}_0(z)-\tau(z) 
        = &\, \frac{1}{n}\sum_{i \in D_2^n}\hat{\mu}(\bX_i, z) - \MP[\hat{\mu}(\bX, z) ] + \int \hat{\mu}(\bx, z) -{\mu}(\bx, z) d \MP(\bx) \\
        = &\, \frac{1}{n}\sum_{i \in D_2^n}\hat{\mu}(\bX_i, z) - \MP[\hat{\mu}(\bX, z) ] + \int (\hat{\mu}(\bx, z) -{\mu}(\bx, z))\frac{f(z)}{\pi(z\mid \bx)} d \MP(\bx \mid z).
    \end{aligned}
    \]
    Plug into the conditional bias term above we have
    \[
    \begin{aligned}
        &\,\ME\left [\hat{\xi}(\bO)-\xi(\bO) \mid D, Z=z \right]\\
        = &\, - \ME \left[(\hat{\mu}(\bX, Z) -{\mu}(\bX, Z)) \left( \frac{\hat{f}(Z)}{\hat{\pi}(Z \mid \bX)}-\frac{f(Z)}{\pi(Z \mid \bX)} \right) \mid D, Z=z\right] + \frac{1}{n}\sum_{i \in D_2^n}\hat{\mu}(\bX_i, z) - \MP[\hat{\mu}(\bX, z) ]\\
        = &\, - \ME \left[ \frac{(\hat{\mu}(\bX, Z) -{\mu}(\bX, Z)) \left( \hat{f}(Z)-f(Z) \right)}{\hat{\pi}(Z \mid \bX)} \mid D, Z=z\right]\\
        &\, -\ME \left[ (\hat{\mu}(\bX, Z) -{\mu}(\bX, Z)) \left( \frac{1}{\hat{\pi}(Z \mid \bX)}-\frac{1}{{\pi}(Z \mid \bX)} \right)f(Z) \mid D, Z=z\right]\\
        &\, + \frac{1}{n}\sum_{i \in D_2^n}\hat{\mu}(\bX_i, z) - \MP[\hat{\mu}(\bX, z) ]
    \end{aligned}
    \]
    Plug this formula of conditional bias into $R_2$, we have
    \begin{equation}\label{eq:R2-decomposition}
        \begin{aligned}
        &\,\MP \left[\bg_h(Z-z_0)K_h(Z-z_0)\left(\hat{\xi}(\bO) - \xi(\bO)\right) \right] \\
        = &\,-\int \bg_h(z-z_0)K_h(z-z_0) \int \frac{(\hat{\mu}(\bx, z) -{\mu}(\bx, z))}{\hat{\pi}(z \mid \bx)} (\hat{f}(z)-f(z)) d \MP(\bx \mid z) f(z) dz\\
        &\, -\ME_{\bX,Z} \left[\bg_h(Z-z_0)K_h(Z-z_0)( \hat{\mu} (\bX, Z) -{\mu}(\bX, Z)) \left( \frac{1}{\hat{\pi}(Z \mid \bX)}-\frac{1}{{\pi}(Z \mid \bX)} \right)f(Z) \right] \\
        &\, + (\MP_n - \MP) \int \bg_h(z-z_0)K_h(z-z_0) \hat{\mu}(\bX,z) f(z) dz,
    \end{aligned}
    \end{equation}
    where the sample average in the last equation is taken over $D_2^n$. For the first term in \eqref{eq:R2-decomposition} we have
    \[
    \begin{aligned}
        &\,\int \bg_h(z-z_0)K_h(z-z_0) \int \frac{(\hat{\mu}(\bx, z) -{\mu}(\bx, z))}{\hat{\pi}(z \mid \bx)} (\hat{f}(z)-f(z)) d \MP(\bx \mid z) f(z) dz\\
        = &\, \int \bg_h(z-z_0)K_h(z-z_0)  \frac{(\hat{\mu}(\bx, z) -{\mu}(\bx, z))}{\hat{\pi}(z \mid \bx)} (\hat{f}(z)-f(z)) d \MP(\bx, z)\\
        = &\, \int \bg_h(z-z_0)K_h(z-z_0)  \frac{(\hat{\mu}(\bx, z) -{\mu}(\bx, z))}{\hat{\pi}(z \mid \bx)} \left(\frac{1}{n} \sum_{i \in D_2^n} \hat{\pi}(z \mid \bX_i) - \ME_{\bX}[\hat{\pi}(z\mid \bX)]  \right) d \MP(\bx, z)\\
        + &\, \int \bg_h(z-z_0)K_h(z-z_0)  \frac{(\hat{\mu}(\bx, z) -{\mu}(\bx, z))}{\hat{\pi}(z \mid \bx)} \left( \ME_{\bX}[\hat{\pi}(z\mid \bX)-\pi(z \mid \bX)] \right) d \MP(\bx, z)\\
    \end{aligned}
    \]
    By Cheybeshev's inequality, 
    \[
    \begin{aligned}
        &\,\ME_{D_2^n}  \left\{\left[ (\MP_n - \MP )\left(\int g_{h,j}(z-z_0)K_h(z-z_0)  \frac{(\hat{\mu}(\bx, z) -{\mu}(\bx, z))}{\hat{\pi}(z \mid \bx)} \hat{\pi}(z \mid \bX)d \MP(\bx, z)\right)\right]^2 \right\}\\
        \leq &\, \frac{1}{n} \ME_{\bX}\left [\left(\int g_{h,j}(z-z_0)K_h(z-z_0)  \frac{(\hat{\mu}(\bx, z) -{\mu}(\bx, z))}{\hat{\pi}(z \mid \bx)} \hat{\pi}(z \mid \bX)d \MP(\bx, z)\right)^2 \right]\\
        \lesssim &\, \frac{1}{n},
    \end{aligned}
    \]
    where the last inequality follows from
    \[
    \int \left|g_{h,j}(z-z_0)K_h(z-z_0)  \frac{(\hat{\mu}(\bx, z) -{\mu}(\bx, z))}{\hat{\pi}(z \mid \bx)} \hat{\pi}(z \mid \bX) \right| d \MP(\bx, z) \lesssim \int |u|^{j-1} K(u) f(z_0 + hu) dz \lesssim 1.
    \]
    This implies 
    \[
    \int \bg_h(z-z_0)K_h(z-z_0)  \frac{(\hat{\mu}(\bx, z) -{\mu}(\bx, z))}{\hat{\pi}(z \mid \bx)} \left(\frac{1}{n} \sum_{i \in D_2^n} \hat{\pi}(z \mid \bX_i) - \ME_{\bX}[\hat{\pi}(z\mid \bX)]  \right) d \MP(\bx, z) = O_{\MP}\left(\frac{1}{\sqrt{n}} \right).
    \]
    For the other term, note that
    \[
    \begin{aligned}
    &\, \ME_D \left [\int \left| g_{h,j}(z-z_0)K_h(z-z_0)  \frac{(\hat{\mu}(\bx, z) -{\mu}(\bx, z))}{\hat{\pi}(z \mid \bx)} \left( \ME_{\bX}[\hat{\pi}(z\mid \bX)-\pi(z \mid \bX)] \right)\right| d \MP(\bx, z)\right]\\
    \leq &\, \int  \left| g_{h,j}(z-z_0)K_h(z-z_0) \right| \ME_D\left[|\hat{\mu}(\bx, z) -{\mu}(\bx, z)| \left( \ME_{\bX}|\hat{\pi}(z\mid \bX)-\pi(z \mid \bX)| \right)\right] d \MP(\bx, z)\\
    \leq &\, \int  \left| g_{h,j}(z-z_0)K_h(z-z_0) \right| \sqrt{\ME_D(\hat{\mu}(\bx, z) -{\mu}(\bx, z))^2} \sqrt{ \ME_{D} \ME_{\bX}(\hat{\pi}(z\mid \bX)-\pi(z \mid \bX))^2 } d \MP(\bx, z)\\
    \lesssim  &\, \sup_{|z-z_0|\leq h} \sqrt{ \ME_{D} \ME_{\bX}(\hat{\pi}(z\mid \bX)-\pi(z \mid \bX))^2 }  \int | g_{h,j}(z-z_0)K_h(z-z_0) | \int \sqrt{\ME_D(\hat{\mu}(\bx, z) -{\mu}(\bx, z))^2} d \MP(\bx) dz \\
    \lesssim &\,\sup_{|z-z_0|\leq h} \sqrt{ \ME_{D} \ME_{\bX}(\hat{\pi}(z\mid \bX)-\pi(z \mid \bX))^2 } \sup_{|z-z_0|\leq h} \sqrt{ \ME_{D} \ME_{\bX}(\hat{\mu}(\bX,z)-\mu(\bX,z))^2 } \int | g_{h,j}(z-z_0)K_h(z-z_0) | dz\\
    \lesssim &\, \sup_{|z-z_0|\leq h} \sqrt{ \ME_{D} \ME_{\bX}(\hat{\pi}(z\mid \bX)-\pi(z \mid \bX))^2 } \sup_{|z-z_0|\leq h} \sqrt{ \ME_{D} \ME_{\bX}(\hat{\mu}(\bX,z)-\mu(\bX,z))^2 },
    \end{aligned}
    \]
    where we apply Cauchy-Schwarz's inequality. Thus we have
    \[
    \begin{aligned}
       &\,\int \bg_h(z-z_0)K_h(z-z_0)  \frac{(\hat{\mu}(\bx, z) -{\mu}(\bx, z))}{\hat{\pi}(z \mid \bx)} \left( \ME_{\bX}[\hat{\pi}(z\mid \bX)-\pi(z \mid \bX)] \right) d \MP(\bx, z)\\
       = &\,  O_{\MP}\left( \sup_{|z-z_0|\leq h} \sqrt{ \ME_{D} \ME_{\bX}(\hat{\pi}(z\mid \bX)-\pi(z \mid \bX))^2 } \sup_{|z-z_0|\leq h} \sqrt{ \ME_{D} \ME_{\bX}(\hat{\mu}(\bX,z)-\mu(\bX,z))^2 } \right) .
    \end{aligned}
    \]
    \[
    \begin{aligned}
        &\,\int \bg_h(z-z_0)K_h(z-z_0) \int \frac{(\hat{\mu}(\bx, z) -{\mu}(\bx, z))}{\hat{\pi}(z \mid \bx)} (\hat{f}(z)-f(z)) d \MP(\bx \mid z) f(z) dz \\
        = &\, O_{\MP}\left( \frac{1}{\sqrt{n}} +\sup_{|z-z_0|\leq h} \sqrt{ \ME_{D} \ME_{\bX}(\hat{\pi}(z\mid \bX)-\pi(z \mid \bX))^2 } \sup_{|z-z_0|\leq h} \sqrt{ \ME_{D} \ME_{\bX}(\hat{\mu}(\bX,z)-\mu(\bX,z))^2 } \right)
    \end{aligned}
    \]
    For the second term in \eqref{eq:R2-decomposition} we have
    \[
    \begin{aligned}
        &\,\ME_D \left\{ \ME_{\bX,Z} \left[|g_{h,j}(Z-z_0)K_h(Z-z_0)|\left|( \hat{\mu} (\bX, Z) -{\mu}(\bX, Z)) \left( \frac{1}{\hat{\pi}(Z \mid \bX)}-\frac{1}{{\pi}(Z \mid \bX)} \right) \right|f(Z) \right] \right\}\\
        \lesssim &\, \ME_{\bX,Z} \left[|g_{h,j}(Z-z_0)K_h(Z-z_0)| \sqrt{\ME_D[(\hat{\mu}(\bX,Z) -\mu(\bX,Z))^2] \ME_D[(\hat{\pi}(Z\mid \bX) -\pi(Z\mid \bX))^2]}  \right]\\
        \lesssim &\, \int |g_{h,j}(z-z_0)K_h(z-z_0)| \int \sqrt{\ME_D[(\hat{\mu}(\bx,z) -\mu(\bx,z))^2] \ME_D[(\hat{\pi}(z\mid \bx) -\pi(z\mid \bx))^2]} d \MP(\bx) dz \\
        \lesssim &\, \sup_{|z-z_0|\leq h} \sqrt{ \ME_{D} \ME_{\bX}(\hat{\pi}(z\mid \bX)-\pi(z \mid \bX))^2 } \sup_{|z-z_0|\leq h} \sqrt{ \ME_{D} \ME_{\bX}(\hat{\mu}(\bX,z)-\mu(\bX,z))^2 },
    \end{aligned}
    \]
    which implies
    \[
    \begin{aligned}
    &\,\ME_{\bX,Z} \left[\bg_h(Z-z_0)K_h(Z-z_0)( \hat{\mu} (\bX, Z) -{\mu}(\bX, Z)) \left( \frac{1}{\hat{\pi}(Z \mid \bX)}-\frac{1}{{\pi}(Z \mid \bX)} \right)f(Z) \right] \\
    = &\, O_{\MP}\left( \sup_{|z-z_0|\leq h} \sqrt{ \ME_{D} \ME_{\bX}(\hat{\pi}(z\mid \bX)-\pi(z \mid \bX))^2 } \sup_{|z-z_0|\leq h} \sqrt{ \ME_{D} \ME_{\bX}(\hat{\mu}(\bX,z)-\mu(\bX,z))^2 }\right)
    \end{aligned}
    \]
    By Cheybeshev's inequality one can similarly show the third term
    \[
    (\MP_n - \MP) \int \bg_h(z-z_0)K_h(z-z_0) \hat{\mu}(\bX,z) f(z) dz = O_{\MP}\left( \frac{1}{\sqrt{n}} \right).
    \]
    We conclude that
    \[
    \begin{aligned}
    &\,\MP \left[\bg_h(Z-z_0)K_h(Z-z_0)\left(\hat{\xi}(\bO) - \xi(\bO)\right) \right] \\
    =  &\, O_{\MP}\left( \frac{1}{\sqrt{n}} +\sup_{|z-z_0|\leq h} \sqrt{ \ME_{D} \ME_{\bX}(\hat{\pi}(z\mid \bX)-\pi(z \mid \bX))^2 } \sup_{|z-z_0|\leq h} \sqrt{ \ME_{D} \ME_{\bX}(\hat{\mu}(\bX,z)-\mu(\bX,z))^2 } \right),
    \end{aligned}
    \]
    \[
    R_2 = \frac{1}{\sqrt{nh^2} }+ \frac{1}{h}\sup_{|z-z_0|\leq h} \sqrt{ \ME_{D} \ME_{\bX}(\hat{\pi}(z\mid \bX)-\pi(z \mid \bX))^2 } \sup_{|z-z_0|\leq h} \sqrt{ \ME_{D} \ME_{\bX}(\hat{\mu}(\bX,z)-\mu(\bX,z))^2 }.
    \]
\end{proof}

\subsection{Proof of Lemma \ref{lemma:CLT-ratio}}
\begin{proof}
    \textbf{Case 1}: $a_n/b_n \rightarrow \infty$. The idea is that $U_n$ has a faster rate and the final rate is dominated by $V_n$. Rewrite 
    \[
    \frac{U_n}{V_n}-\frac{u_n}{v_n}=\frac{U_n-u_n}{V_n}+u_n\left(\frac{1}{V_n}-\frac{1}{v_n} \right).
    \]
    Since $1/V_n \stackrel{P}{\rightarrow}1/\theta_V$, we have $1/V_n = O_{\MP}(1)$, this together with $U_n-u_n=O_{\MP}(1/a_n)$ implies
    \[
    b_n \frac{U_n-u_n}{V_n} = b_n O_{\MP}(1)O_{\MP}(1/a_n)=O_{\MP}(b_n/a_n)=o_{\MP}(1).
    \]
    For the second term, by delta method we have
    \begin{equation}\label{eq:delta}
        b_n\left(\frac{1}{V_n}-\frac{1}{v_n} \right) \stackrel{d}{\rightarrow} N(0,\sigma_V^2/\theta_V^4),
    \end{equation}
    then apply Slutsky's theorem we obtain 
    \[
    u_nb_n\left(\frac{1}{V_n}-\frac{1}{v_n} \right)\stackrel{d}{\rightarrow} N(0,\theta_U^2\sigma_V^2/\theta_V^4).
    \]
    \[
    b_n \left( \frac{U_n}{V_n}-\frac{u_n}{v_n}\right) \rightarrow N(0,\theta_U^2\sigma_V^2/\theta_V^4).
    \]
    \textbf{Case 2}: $a_n/b_n \rightarrow 0$. Now the final rate is dominated by $U_n$. We write
    \[
    a_n \left( \frac{U_n}{V_n}-\frac{u_n}{v_n} \right)=a_n \left( \frac{U_n}{V_n}-\frac{u_n}{V_n} \right) + a_n\left( \frac{u_n}{V_n}-\frac{u_n}{v_n} \right).
    \]
    By equation \eqref{eq:delta} we have $1/V_n-1/v_n = O_{\MP}(1/b_n)$, which implies
    \[
    a_n\left( \frac{u_n}{V_n}-\frac{u_n}{v_n} \right) = a_n O_{\MP}(1)O_{\MP}(1/b_n) = O_{\MP}(a_n/b_n)=o_{\MP}(1).
    \]
    For the first term, by Slutsky's theorem we have
    \[
    a_n \left( \frac{U_n}{V_n}-\frac{u_n}{V_n} \right)\stackrel{d}{\rightarrow} N(0, \sigma_U^2/\theta_V^2).
    \]
    Thus we have
    \[
    a_n \left( \frac{U_n}{V_n}-\frac{u_n}{v_n} \right)\stackrel{d}{\rightarrow} N(0, \sigma_U^2/\theta_V^2).
    \]
    \textbf{Case 3}: $a_n=b_n$ (note that the constants can be absorbed into the variance $\sigma_U^2, \sigma_V^2$ so we only need to consider this case here). In this case the result follows from the uniform delta method \citep{van2000asymptotic}[Section 3.4].
\end{proof}

\subsection{Proof of Lemma \ref{lemma:bounds-remainder}}
\begin{proof}
    By Lemma 2 in \cite{kennedy2020sharp} we have
    \[
    (\MP_n-\MP)[\hat{\varphi}_h(\bO; z_0)-\varphi_h(\bO;z_0)] = O_{\MP}\left( \frac{\|\hat{\varphi}_h(\bO; z_0)-\varphi_h(\bO;z_0)\|_2}{\sqrt{n}}\right)
    \]
    By direct calculations,
    \[
    \begin{aligned}
        &\,\hat{\varphi}_h(\bO; z_0)-\varphi_h(\bO;z_0)\\
        = &\, \frac{K_h'(Z-z_0)(Y-\mu(\bX,Z))(\hat{\pi}(Z\mid \bX)-\pi(Z \mid \bX))}{\hat{\pi}(Z\mid \bX)\pi(Z \mid \bX)} \\
        &\,+ \frac{K_h'(Z-z_0)(\hat{\mu}(\bX,Z) - \mu(\bX,Z))}{\hat{\pi}(Z \mid \bX)} - \int (\hat{\mu}(\bX,z) - \mu(\bX,z))K_h'(z-z_0) dz.
    \end{aligned}
    \]
    We have
    \[
    \begin{aligned}
        &\,\ME_D \left[\left\|\frac{K_h'(Z-z_0)(\hat{\mu}(\bX,Z) - \mu(\bX,Z))}{\hat{\pi}(Z \mid \bX)}\right\|_2^2 \right]\\
        \lesssim &\, \ME_D \left[ \ME_{\bX,Z} \left((K_h'(Z-z_0))^2(\hat{\mu}(\bX,Z) - \mu(\bX,Z))^2\right) \right]\\
        = &\,  \ME_{\bX,Z}\left[ (K_h'(Z-z_0))^2 \ME_D \left(\hat{\mu}(\bX,Z) - \mu(\bX,Z)\right)^2 \right]\\
        = &\, \int \int (K_h'(z-z_0))^2\ME_D \left(\hat{\mu}(\bx,z) - \mu(\bx,z)\right)^2 \pi(z \mid \bx) dz d \MP(\bx)\\
        \lesssim &\, \int (K_h'(z-z_0))^2 \int \ME_D \left(\hat{\mu}(\bx,z) - \mu(\bx,z)\right)^2  d \MP(\bx) dz\\
        \leq &\, \sup_{|z-z_0|\leq h} \ME_{\bX}[\ME_D \left(\hat{\mu}(\bX,z) - \mu(\bX,z)\right)^2] \int (K_h'(z-z_0))^2 dz \\
        = &\, \frac{1}{h^3}\sup_{|z-z_0|\leq h} \ME_{\bX}[\ME_D \left(\hat{\mu}(\bX,z) - \mu(\bX,z)\right)^2] \int  \left( K'(u)\right)^2 du,
    \end{aligned}
    \]
    where the first inequality follows from positivity of $\hat{\pi}$ and the equation follows from Fubini's Theorem. The second inequality follows from $\pi \leq C$. So we have
    \[
    \left\|\frac{K_h'(Z-z_0)(\hat{\mu}(\bX,Z) - \mu(\bX,Z))}{\hat{\pi}(Z \mid \bX)}\right\|_2 = O_{\MP} \left(\frac{1}{\sqrt{h^3}} \sup_{|z-z_0|\leq h} \sqrt{\ME_{\bX}[\ME_D \left(\hat{\mu}(\bX,z) - \mu(\bX,z)\right)^2]} \right)
    \]
    Similarly one could show
    \[
    \left\|\frac{K_h'(Z-z_0)(Y-\mu(\bX,Z))(\hat{\pi}(Z\mid \bX)-\pi(Z \mid \bX))}{\hat{\pi}(Z\mid \bX)\pi(Z \mid \bX)} \right\|_2 = O_{\MP} \left(\frac{1}{\sqrt{h^3}} \sup_{|z-z_0|\leq h} \sqrt{\ME_{\bX}[\ME_D \left(\hat{\pi}(z\mid\bX) - \pi(z, \bX)\right)^2]} \right).
    \]
    For the third term, by Generalized Minkowski inequality we have
    \[
    \begin{aligned}
        &\,\left\|\int (\hat{\mu}(\bX,z) - \mu(\bX,z))K_h'(z-z_0) \right\|_2\\
        =&\,\left[\ME_{\bX} \left(\int (\hat{\mu}(\bX,z) -\mu(\bX,z))K_h'(z-z_0) dz \right)^2 \right]^{1/2}\\
        \leq &\, \int \left(\int (\hat{\mu}(\bx,z) -\mu(\bx,z))^2(K_h'(z-z_0))^2 d \MP(\bx)\right)^{1/2} dz 
    \end{aligned}
    \]
    By Cauchy Schwarz inequality and Fubini's theorem, we have
    \[
    \begin{aligned}
        &\,\ME_D\left[ \int  \left(\int (\hat{\mu}(\bx,z) -\mu(\bx,z))^2(K_h'(z-z_0))^2 d \MP(\bx)\right)^{1/2} dz\right]\\
        = &\, \ME_D\left[ \int|K_h'(z-z_0)| \left(\int (\hat{\mu}(\bx,z) -\mu(\bx,z))^2 d \MP(\bx)\right)^{1/2} dz  \right]\\
        = &\,  \int |K_h'(z-z_0)| \ME_D \left[\left(\int (\hat{\mu}(\bx,z) -\mu(\bx,z))^2 d \MP(\bx)\right)^{1/2}\right] dz \\
        \leq &\, \int |K_h'(z-z_0)| \sqrt{\ME_D \left[\int (\hat{\mu}(\bx,z) -\mu(\bx,z))^2 d \MP(\bx)\right]} dz \\
        \leq &\, \sup_{|z-z_0|\leq h} \sqrt{\ME_{\bX}[\ME_D \left(\hat{\mu}(\bX,z) - \mu(\bX,z)\right)^2]} \int |K_h'(z-z_0)| dz \\
        = &\, \frac{1}{h} \sup_{|z-z_0|\leq h} \sqrt{\ME_{\bX}[\ME_D \left(\hat{\mu}(\bX,z) - \mu(\bX,z)\right)^2]} \int |K'(u)| du  
    \end{aligned}
    \]
    Hence we have
    \[
    \left\|\int (\hat{\mu}(\bX,z) - \mu(\bX,z))K_h'(z-z_0) \right\|_2 = O_{\MP} \left(\frac{1}{h} \sup_{|z-z_0|\leq h} \sqrt{\ME_{\bX}[\ME_D \left(\hat{\mu}(\bX,z) - \mu(\bX,z)\right)^2]}\right).
    \]
    So the empirical process term can be bounded as 
    \[
    \begin{aligned}
        &\,(\MP_n-\MP)[\hat{\varphi}_h(\bO; z_0)-\varphi_h(\bO;z_0)] \\
        =&\, O_{\MP} \left( \frac{1}{\sqrt{nh^3}} \max \left\{ \sup_{|z-z_0|\leq h} \sqrt{\ME_{\bX}[\ME_D \left(\hat{\mu}(\bX,z) - \mu(\bX,z)\right)^2]}, \sup_{|z-z_0|\leq h} \sqrt{\ME_{\bX}[\ME_D \left(\hat{\pi}(z \mid \bX) - \pi(z \mid \bX)\right)^2]}\right\} \right)\\
    \end{aligned}
    \]
    
\end{proof}

\subsection{Proof of Lemma \ref{lemma:lindeberg}}
\begin{proof}[Proof of Lemma \ref{lemma:lindeberg}]
    Since $L_n / B_n \rightarrow 0$, for any $\tau >0$ we can find $n_0 \in \mathbb{N}_+$ such that for all $n \geq n_0$ we have $2L_n / B_n < \tau$. Then note that for all $n \geq n_0$,
    \[
    \frac{\max_{1 \leq k \leq k_n}|X_{nk}-\ME[X_{nk}]|}{B_n} \leq \frac{2L_n}{B_n} < \tau,
    \]
    which implies
    \[
    \left\{| X_{nk}-\ME[X_{nk}] | \geq \tau B_n\right\}=\varnothing, \quad k=1, \cdots, k_{n}
    \]
    and thus 
    \[
    \frac{1}{B_n^2} \sum_{k=1}^{k_n} \mathbf{E}\left[\left(X_{nk}-\ME[X_{nk}]\right)^2 I\left(\left|X_{nk}-\ME[X_{nk}]\right| \geq \tau B_n\right)\right]=0
    \]
    when $n$ is sufficiently large.
\end{proof}

\section{Proof of Main Results}

\subsection{Proof of Theorem \ref{thm:LP-boundary}}
\begin{proof}
    To prove the asymptotic expansion, similar to the proof of Lemma \ref{lemma:LP-oracle}, we can write
    \[
    \begin{aligned}
        \hat{\tau}(z_0) - \tau(z_0) = &\,\tilde{\tau}(z_0) - \tau(z_0) + \hat{\tau}(z_0) - \tilde{\tau}(z_0)\\
        = &\, \tilde{\tau}(z_0) - \tau(z_0) + \be_1^{\top} \hat{\bD}_{hz_0}^{-1} \MP_n \left[\bg_h(Z-z_0)K_h(Z-z_0)\left(\hat{\xi}(\bO) - \xi(\bO)\right) \right]\\
        = &\, \tilde{\tau}(z_0) - \tau(z_0) + \be_1^{\top} \hat{\bD}_{hz_0}^{-1} (\MP_n-\MP) \left[\bg_h(Z-z_0)K_h(Z-z_0)\left(\hat{\xi}(\bO) - \xi(\bO)\right) \right] \\
        &\,+ \be_1^{\top} \hat{\bD}_{hz_0}^{-1} \MP \left[\bg_h(Z-z_0)K_h(Z-z_0)\left(\hat{\xi}(\bO) - \xi(\bO)\right) \right].
    \end{aligned}
    \]
    The proof then follows from the same calculations as in that of Lemma \ref{lemma:LP-oracle}, with $z_0$ being a point on the boundary instead an interior point. For example, the same proof of Theorem 3 in \cite{zeng2024continuous} shows 
    \[
    \hat{D}_{hz_0,j\ell} \stackrel{P}{\rightarrow} \ME \left[ \left(\frac{Z-z_0}{h} \right)^{j+\ell} K_h(Z-z_0) \right].
    \]
    When $z_0 = ch$ lies on the boundary, we have (assume $n$ is sufficiently large so that $1/h-c>1$)
    \[
    \begin{aligned}
        &\,\ME \left[ \left(\frac{Z-z_0}{h} \right)^{j+\ell} K_h(Z-z_0) \right]\\
        = &\, \int_0^1 \left(\frac{z-z_0}{h} \right)^{j+\ell} K_h(z-z_0) f(z ) dz\\
        = &\, \int_{-c}^{1/h-c} u^{j+\ell} K(u) f(z_0 +hu ) du\\
        =&\, \int_{-c}^{1} u^{j+\ell} K(u) f(z_0 +hu ) du\\
        \rightarrow &\, f(z_0) \int_{-c}^{1} u^{j+\ell} K(u)  du.
    \end{aligned}
    \]
    Note that when $z_0$ is an interior point the limit of $\hat{D}_{hz_0,j\ell}$ is $f(z_0) \int_{-1}^{1} u^{j+\ell} K(u)  du$. One can proceed similarly as in Lemma \ref{lemma:LP-oracle} to bounding the empirical process term and the conditional bias. For example, to bound the empirical process term, for the first term in \eqref{eq:xi-error} we have
    \[
    \begin{aligned}
        &\, \ME_D \left[ \MP\left(g_{h,j}^2(Z-z_0)K_h^2(Z-z_0)\frac{(Y-\hat{\mu}(\bX,Z))^2(\hat{f}(Z)-f(Z))^2}{\hat{\pi}^2(Z \mid \bX)}\right)\right]\\
        \lesssim &\, \frac{1}{h} \int_{-c}^1 u^{2(j-1)} K^2(u) \ME_D \left[(\hat{f}(z_0+hu)-f(z_0+hu))^2\right] f(z_0+hu) du \\
        \leq &\, \frac{1}{h} \sup_{0 \leq z \leq z_0+h } \ME_D \left[(\hat{f}(z)-f(z))^2\right] \int_{-c}^1 u^{2(j-1)} K^2(u) f(z_0+hu) du \\
        \lesssim &\, \frac{1}{h}\sup_{0 \leq z \leq z_0+h } \ME_D \left[(\hat{f}(z)-f(z))^2\right].
    \end{aligned}
    \]
    Note that the range of $z$ is $[0,z_0+h]$. The remaining proof is similar and omitted. The final rate follows from Theorem 3.2 of \cite{fan2018local}.
\end{proof}

\subsection{Proof of Theorem \ref{thm:dose-response-smooth}}
\begin{proof}
    Under Assumption \ref{ass:smoothness}, the estimation error in Theorem \ref{thm:LP-boundary} is given by
\[
\hat{\tau}(z_0)-\tau(z_0)= O_{\MP} \left(h^{\gamma} + \frac{1}{\sqrt{nh}} + n^{-\left( \frac{1}{2+\frac{1}{\gamma}+\frac{d}{\beta}}+ \frac{1}{2+\frac{d+1}{\alpha}} \right)} \right).
\]
We can select $h$ to minimize the estimation error in Theorem \ref{thm:LP-boundary}. The results in two different smoothing regimes are summarized as follows:

\noindent \textbf{Case 1 :} The oracle regime 
\begin{equation*}
    \frac{d/\beta}{(2+1/\gamma)(2+1/\gamma+d/\beta)} \leq \frac{\alpha}{2\alpha+d+1}
\end{equation*}
or equivalently,
\[
\frac{1}{2+1/\gamma + d/\beta} + \frac{1}{2+(d+1)/\alpha} \geq \frac{\gamma}{2\gamma+1}.
\]
In this regime, the nuisance functions can be estimated at sufficiently fast rates and we can set $h \asymp n^{-\frac{1}{2\gamma+1}}$ to achieve the oracle rate for estimating a $\gamma$-smooth function:
\[
\hat{\tau}(z_0)-\tau(z_0)=O_{\MP}\left(  n^{-\frac{\gamma}{2\gamma+1}} \right).
\]
\textbf{Case 2:} The alternative regime 
\[
\frac{d/\beta}{(2+1/\gamma)(2+1/\gamma+d/\beta)} > \frac{\alpha}{2\alpha+d+1}.
\]
In this regime, the nuisance estimation error dominates and the final rate is 
\[
\hat{\tau}(z_0)-\tau(z_0) = O_{\MP} \left( n^{-\left( \frac{1}{2+\frac{1}{\gamma}+\frac{d}{\beta}}+ \frac{1}{2+\frac{d+1}{\alpha}} \right)} \right).
\]
\end{proof}

\subsection{Proof of Theorem \ref{thm:deriv-smooth}}\label{appendix:proof-deriv-smooth}
\begin{proof}
    Following the proof of \cite{Tsybakov2009}[Exercise 1.4], when $\tau$ is a $\gamma$-smooth function and $p=\lfloor \gamma \rfloor$, the MSE of the oracle estimator $\tilde{\theta}(z_0)$ can be bounded as $h^{2(\gamma-1)}+\frac{1}{nh^3}$ (under regular conditions specified there for local polynomial estimators), which implies
\[
\tilde{\theta}(z_0) -\theta(z_0) = O_{\MP}\left(h^{\gamma-1}+ \frac{1}{\sqrt{nh^3}} \right).
\]
Since the estimates $\hat{\mu}, \hat{\pi}$ are consistent
\[
\max \left\{ \sup_{|z-z_0|\leq h} \sqrt{\ME_{\bX}[\ME_D \left(\hat{\mu}(\bX,z) - \mu(\bX,z)\right)^2]}, \sup_{|z-z_0|\leq h} \sqrt{\ME_{\bX}[\ME_D \left(\hat{\pi}(z \mid \bX) - \pi(z \mid \bX)\right)^2]}\right\} \rightarrow 0.
\]
Lemma \ref{lemma:LP-oracle} then implies
\[
\begin{aligned}
    &\,\hat{\theta}(z_0) - \theta(z_0)\\
    = &\, O_{\MP}\left(h^{\gamma-1}+ \frac{1}{\sqrt{nh^3}}+ \frac{1}{h}\sup_{|z-z_0|\leq h} \sqrt{ \ME_{D} \ME_{\bX}(\hat{\pi}(z\mid \bX)-\pi(z \mid \bX))^2 } \sup_{|z-z_0|\leq h} \sqrt{ \ME_{D} \ME_{\bX}(\hat{\mu}(\bX,z)-\mu(\bX,z))^2 }\right)
\end{aligned}
\]
Under Assumption \ref{ass:smoothness}, the estimation error of $\hat{\theta}(z_0)$ is bounded as:
\begin{equation}\label{eq:LP-deriv-error}
    h^{\gamma-1} + \frac{1}{h} n^{-\left( \frac{1}{2+\frac{1}{\gamma}+\frac{d}{\beta}}+ \frac{1}{2+\frac{d+1}{\alpha}} \right)}+ \frac{1}{\sqrt{nh^3}}.
\end{equation}
As in the rate analysis in Section \ref{sec:dose-res-boundary}, the optimal choice of the bandwidth $h$ and the corresponding rate depend on the regime of the smoothness parameters $\alpha, \beta, \gamma$: 

\textbf{Case 1}: The oracle regime
\begin{equation*}
    \frac{d/\beta}{(2+1/\gamma)(2+1/\gamma+d/\beta)} \leq \frac{\alpha}{2\alpha+d+1}
\end{equation*}
or equivalently,
\[
\frac{1}{2+1/\gamma + d/\beta} + \frac{1}{2+(d+1)/\alpha} - \frac{1}{2\gamma+1} \geq \frac{\gamma-1}{2\gamma+1}.
\]
In this regime, the nuisance functions can be estimated at sufficiently fast rates, allowing us to balance $h^{\gamma-1}$ with $1/\sqrt{nh^3}$ by setting $h \asymp n^{-\frac{1}{2\gamma+1}}$. This yields:
\[
\hat{\theta}(z_0)-\theta(z_0)= O_{\MP}\left( n^{-\frac{\gamma-1}{2\gamma+1}} + n^{-\left(\frac{1}{2+1/\gamma + d/\beta} + \frac{1}{2+(d+1)/\alpha} - \frac{1}{2\gamma+1}\right)} \right)= O_{\MP}\left(n^{-\frac{\gamma-1}{2\gamma+1}}\right),
\]
which matches the rate for estimating the first-order derivative of a $\gamma$-smooth function \citep{Tsybakov2009}.

\textbf{Case 2}: The alternative regime
\[
\frac{d/\beta}{(2+1/\gamma)(2+1/\gamma+d/\beta)} > \frac{\alpha}{2\alpha+d+1}.
\]
In this regime, the nuisance estimation error is larger, requiring a larger bandwidth (compared to $h \asymp n^{-\frac{1}{2\gamma+1}}$) to minimize its contribution in \eqref{eq:LP-deriv-error}. A larger bandwidth reduces the variance term $1/\sqrt{nh^3}$, which then decays faster than the bias term $h^{\gamma-1}$. To balance these terms, we solve:
\[
h^{\gamma-1} \asymp \frac{1}{h} n^{-\left( \frac{1}{2+\frac{1}{\gamma}+\frac{d}{\beta}}+ \frac{1}{2+\frac{d+1}{\alpha}} \right)},
\]
or equivalently,
\[
h \asymp  n^{- \frac{1}{\gamma}\left( \frac{1}{2+\frac{1}{\gamma}+\frac{d}{\beta}}+ \frac{1}{2+\frac{d+1}{\alpha}} \right)},
\]
which yields the final rate for $\hat{\theta}(z_0)$ as
\[
\hat{\theta}(z_0)-\theta(z_0)=O_{\MP}\left( n^{- \frac{\gamma-1}{\gamma}\left( \frac{1}{2+\frac{1}{\gamma}+\frac{d}{\beta}}+ \frac{1}{2+\frac{d+1}{\alpha}} \right)} \right).
\]
\end{proof}

\subsection{Proof of Theorem \ref{thm:LP-normality}}

\begin{proof}
    The proof mainly follows from that of \cite{sawada2024local}. Note that the condition $nh^{2p+3} = O(1)$ is mainly used to obtain a specific order for the bias term. Following the notation in \cite{sawada2024local}, the upper bound on $h$ is used to derive an asymptotic expansion for $B_{n, j_1 \ldots j_L 2}+B_{n, j_1 \ldots j_L 4}$. Without the upper bound $nh^{2p+3} = O(1)$, we can keep $B_{n, j_1 \ldots j_L 2}+B_{n, j_1 \ldots j_L 4}$ in our analysis and result, which yields a bias term
    \[
    \frac{1}{(p+1)!} \sqrt{nh} \ME \left[ K_h(Z-z_0) \bg_h(Z-z_0)\theta^{(p+1)}(\tilde{Z})(Z-z_0)^{p+1} \right].
    \]
    The proof in \cite{sawada2024local} then yields
    \[
    \begin{aligned}
    &\,\sqrt{n h}\bH \left(\tilde{\bbb}(z_0) - \boldsymbol{\theta}(z_0)\right) - \frac{1}{(p+1)!} \sqrt{nh} \bS_n^{-1}(z_0) \ME \left[ K_h(Z-z_0) \bg_h(Z-z_0)\theta^{(p+1)}(\tilde{Z})(Z-z_0)^{p+1} \right] \\
    &\, \stackrel{d}{\rightarrow} N(0, \sigma^2(z_0)\bV/f(z_0)),
    \end{aligned}
    \]
    where $\boldsymbol{\theta}(z_0) = (\theta(z_0), \theta'(z_0),\dots, \theta^{(p)}(z_0))^{\top}$ and $\hat{\bbb}(z_0)$ is the local polynomial estimator of $\boldsymbol{\theta}(z_0)$ using the oracle pseudo-outcome $\xi(\bO;\pi,\mu)$.
    Our result in \eqref{eq:LP-normality} then follows from taking the second component and Lemma \ref{lemma:LP-oracle}. When $nh^{2p+3} = O(1)$ holds, the analysis in \cite{sawada2024local} shows the leading term of the bias 
    \[
    \frac{1}{(p+1)!h} \be_2^{\top} \bS_n^{-1}(z_0)\ME \left[ K_h(Z-z_0) \bg_h(Z-z_0)\theta^{(p+1)}(\tilde{Z})(Z-z_0)^{p+1} \right]
    \]
    is equal to 
    \[
    \frac{1}{(p+1)!}\theta^{(p+1)}(z_0) \be_2^{\top} \bS^{-1}(\mu_{p+1},\dots, \mu_{2p+1})^{\top}h^{p}.
    \]
\end{proof}

\subsection{Proof of Proposition \ref{prop:smooth-error}}

\begin{proof}
    By definition of $\theta_h$ we have
    \[
    \begin{aligned}
        &\,\theta_h(z_0) - \theta(z_0) \\
        = &\, \ME \left[\int\frac{\partial \mu(\bX, z)}{\partial z} K_h(z-z_0) dz - \left. \frac{\partial \mu(\bX, z)}{\partial z} \right|_{z=z_0} \right]\\
        = &\, \ME\left[\int\left.\frac{\partial \mu(\bX, z)}{\partial z}\right|_{z_0}^z K_h(z-z_0) dz   \right]\\
        = &\, \ME\left[\int\left.\frac{\partial \mu(\bX, z)}{\partial z}\right|_{z_0}^{z_0+hu} K(u) du   \right]
    \end{aligned}
    \]
    where $\left.\frac{\partial \mu(\bX, z)}{\partial z}\right|_{z_1}^{z_2} = \left.\frac{\partial \mu(\bX, z)}{\partial z}\right|_{z=z_2}-\left.\frac{\partial \mu(\bX, z)}{\partial z}\right|_{z=z_1}$ and the last equation follows from change of variables $u = (z-z_0)/h$. By Taylor's expansion we have for some $\tau \in (0,1)$, 
    \[
    \left.\frac{\partial \mu(\bX, z)}{\partial z}\right|_{z_0}^{z_0+hu} = \sum_{j=1}^{\ell-2} \frac{1}{j!} \left.\frac{\partial^{j+1} \mu(\bX,z)}{\partial z^{j+1}}\right|_{z=z_0} (hu)^j + \frac{1}{(\ell-1)!} \left.\frac{\partial^{\ell} \mu(\bX,z)}{\partial z^{\ell}} \right|_{z=z_0+\tau hu} (hu)^{\ell-1}. 
    \]
    Since $K$ is a ($\ell-1$)-th order kernel, we have
    \[
    \begin{aligned}
    &\,\ME\left[\int\left.\frac{\partial \mu(\bX, z)}{\partial z}\right|_{z_0}^{z_0+hu} K(u) du   \right]\\
    = &\, \ME\left[\int\frac{1}{(\ell-1)!} \left.\frac{\partial^{\ell} \mu(\bX,z)}{\partial z^{\ell}} \right|_{z=z_0+\tau hu} (hu)^{\ell-1} K(u) du   \right]\\
    = &\,\ME\left[\int\frac{1}{(\ell-1)!} \left.\frac{\partial^{\ell} \mu(\bX,z)}{\partial z^{\ell}} \right|_{z_0}^{z_0+\tau hu} (hu)^{\ell-1} K(u) du   \right].
    \end{aligned}
    \]
    Thus the approximation error can be bounded as 
    \[
    \begin{aligned}
     &\,|\theta_h(z_0) - \theta(z_0)| \\
      \leq &\,\frac{1}{(\ell-1)!} \ME\left[\int \left| \left.\frac{\partial^{\ell} \mu(\bX,z)}{\partial z^{\ell}} \right|_{z_0}^{z_0+\tau hu} \right| (h|u|)^{\ell-1} |K(u)| du   \right] \\
      \leq &\,\frac{L}{(\ell-1)!} \left[\int |\tau hu|^{\gamma-\ell} (h|u|)^{\ell-1} |K(u)| du   \right] \\
      \leq &\,\frac{Lh^{\gamma-1}}{(\ell-1)!} \left[\int |u|^{\gamma-1} |K(u)| du   \right]
    \end{aligned}
    \]
    
\end{proof}

\subsection{Proof of Proposition \ref{prop:bound-variance}}
    \begin{proof}
    The conditional bias can be directly calculated by
    \[
    \begin{aligned}
      &\,\ME[\hat{\theta}_h(z_0) - \theta_h(z_0)] \\
      = &\, -\ME \left[K_h'(Z-z_0) \frac{Y-\hat{\mu}(\bX,Z)}{\hat{\pi}(Z\mid \bX)} + \int  (\hat{\mu}(\bX, z)-{\mu}(\bX, z))K_h'(z-z_0) dz \right] \\
      = &\, -\ME \left[K_h'(Z-z_0) \frac{\mu(\bX,Z)-\hat{\mu}(\bX,Z)}{\hat{\pi}(Z\mid \bX)} + \int  (\hat{\mu}(\bX, z)-{\mu}(\bX, z))K_h'(z-z_0) dz \right]\\
      = &\, -\ME \left[\int K_h'(z-z_0) \frac{(\mu(\bX,z)-\hat{\mu}(\bX,z)) \pi(z \mid \bX)}{\hat{\pi}(z\mid \bX)}dz + \int  (\hat{\mu}(\bX, z)-{\mu}(\bX, z))K_h'(z-z_0) dz \right] \\
      = &\, -\ME\left [ \int K_h'(z-z_0)  (\hat{\mu}(\bX,z)-{\mu}(\bX,z)) \left( 1-\frac{\pi(z \mid \bX)}{\hat{\pi}(z \mid \bX)}\right) dz\right].
    \end{aligned}
    \]
    By Fubini's theorem, Cauchy Schwarz inequality and positivity assumption, it is bounded by
    \[
    |\ME[\hat{\theta}_h(z_0) - \theta_h(z_0)]| \lesssim \int |K_h'(z-z_0)| \|\hat{\mu}(\cdot, z) - \mu(\cdot, z)\|_2 \|\hat{\pi}(z\mid \cdot) - \pi(z \mid \cdot)\|_2 dz .
    \]
    We use $\ME_D$ to denote the expectation taken w.r.t. the data used to train the nuisance functions and $\ME_{\bX,Z,Y}$ to denote the expectation taken w.r.t. a new data point $(\bX,Z,Y)$ independent of $D$. 
    By Fubini's Theorem and Jenson's inequality, we have
    \[
    \begin{aligned}
        &\,\ME_D \left[\int |K_h'(z-z_0)| \|\hat{\mu}(\cdot, z) - \mu(\cdot, z)\|_2 \|\hat{\pi}(z\mid \cdot) - \pi(z \mid \cdot)\|_2 dz \right]\\
        = &\, \int |K_h'(z-z_0)| \ME_D[\|\hat{\mu}(\cdot, z) - \mu(\cdot, z)\|_2 \|\hat{\pi}(z\mid \cdot) - \pi(z \mid \cdot)\|_2] dz \\
        \leq &\, \int |K_h'(z-z_0)| \sqrt{\ME_D \ME_\bX [(\hat{\mu}(\bX, z) - \mu(\bX, z))^2]} \sqrt{\ME_D \ME_\bX[(\hat{\pi}(z\mid \bX) - \pi(z \mid \bX))^2]} dz \\
        = &\, \int |K_h'(z-z_0)| \sqrt{ \ME_\bX \ME_D [(\hat{\mu}(\bX, z) - \mu(\bX, z))^2]} \sqrt{ \ME_\bX \ME_D[(\hat{\pi}(z\mid \bX) - \pi(z \mid \bX))^2]} dz \\
        \leq &\, \sup_{|z-z_0|\leq h} \sqrt{\ME_\bX\ME_D[(\hat{\pi}(z\mid \bX)-\pi(z\mid \bX))^2]} \sup_{|z-z_0|\leq h} \sqrt{\ME_\bX\ME_D[(\hat{\mu}( \bX,z)-\mu(\bX,z))^2]} \int |K_h'(z-z_0)| dz\\
        = &\, \frac{1}{h}\sup_{|z-z_0|\leq h} \sqrt{\ME_\bX\ME_D[(\hat{\pi}(z\mid \bX)-\pi(z\mid \bX))^2]} \sup_{|z-z_0|\leq h} \sqrt{\ME_\bX\ME_D[(\hat{\mu}( \bX,z)-\mu(\bX,z))^2]} \int |K'(u)| du  .
    \end{aligned}
    \]
    Hence the conditional bias can be bounded as 
    \[
    \begin{aligned}
        &\,\MP[\hat{\varphi}_h(\bO; z_0)-\varphi_h(\bO;z_0)]\\
        = &\, O_{\MP} \left(\frac{1}{h}\sup_{|z-z_0|\leq h} \sqrt{\ME_\bX\ME_D[(\hat{\pi}(z\mid \bX)-\pi(z\mid \bX))^2]} \sup_{|z-z_0|\leq h} \sqrt{\ME_\bX\ME_D[(\hat{\mu}( \bX,z)-\mu(\bX,z))^2]} \right).
    \end{aligned}
    \]
    The conditional variance of $\hat{\theta}_h(z_0)$ is 
    \[
    \begin{aligned}
       &\, \operatorname{Var}\left(\hat{\theta}_h(z_0)\right)\\
       = &\,\frac{1}{n} \operatorname{Var}\left(K_h'(Z-z_0)  \frac{Y-\hat{\mu}(\bX,Z)}{\hat{\pi}(Z\mid \bX)} + \int  \hat{\mu}(\bX, z)K_h'(z-z_0) dz \right) \\
       \leq & \, \frac{2}{n} \left[\operatorname{Var}\left(K_h'(Z-z_0)  \frac{Y-\hat{\mu}(\bX,Z)}{\hat{\pi}(Z\mid \bX)}\right) + \operatorname{Var}\left( \int  \hat{\mu}(\bX, z)K_h'(z-z_0) dz \right) \right] 
    \end{aligned}
    \]
    For the first term we have
    \[
    \begin{aligned}
        &\,\operatorname{Var}\left(K_h'(Z-z_0)  \frac{Y-\hat{\mu}(\bX,Z)}{\hat{\pi}(Z\mid \bX)}\right)\\
        \leq &\, \ME\left[ \left(K_h'(Z-z_0)\right)^2  \frac{(Y-\hat{\mu}(\bX,Z))^2}{\hat{\pi}^2(Z\mid \bX)}\right]\\
        \lesssim &\, \ME\left[ \left(K_h'(Z-z_0)\right)^2\right]\\
        = &\, \int \frac{1}{h^4}\left(K'\left(\frac{Z-z_0}{h}\right)\right)^2 f(z) dz \\
        \lesssim &\, \frac{1}{h^3} \int \left(K'\left(u\right)\right)^2 du \\
        \lesssim &\, \frac{1}{h^3},
    \end{aligned}
    \]
    where the second inequality follows from bounds on the nuisance estimators. For the second term we have
    \[
    \begin{aligned}
        &\, \operatorname{Var}\left( \int  \hat{\mu}(\bX, z)K_h'(z-z_0) dz \right) \\
        \leq &\, \ME \left[\left(\int  \hat{\mu}(\bX, z)K_h'(z-z_0) dz\right)^2\right].
    \end{aligned}
    \]
    Similar calculations show
    \[
    \begin{aligned}
        &\,\left|\int  \hat{\mu}(\bX, z)K_h'(z-z_0) dz \right|\\
        \lesssim &\, \int |K_h'(z-z_0)| dz \\
        = &\, \frac{1}{h} \int |K'(u)|du.
    \end{aligned}
    \]
    Hence we have
    \[
     \operatorname{Var}\left( \int  \hat{\mu}(\bX, z)K_h'(z-z_0) dz \right)\leq \ME \left[\left(\int  \hat{\mu}(\bX, z)K_h'(z-z_0) dz\right)^2\right]\lesssim \frac{1}{h^2}.
    \]
    \[
    \operatorname{Var}\left(\hat{\theta}_h(z_0)\right) \lesssim \frac{1}{nh^3}
    \]
    
    \end{proof}

\subsection{Proof of Theorem \ref{thm:smooth-normality}}
\begin{proof}
    Recall we have the following decomposition of estimation error
    \[
    \begin{aligned}
        \hat{\theta}_h(z_0) - \theta(z_0) = &\,\hat{\theta}_h(z_0) - {\theta}_h(z_0) + {\theta}_h(z_0)-\theta(z_0)\\
        = &\, (\MP_n-\MP) [\varphi_h(\bO;z_0)] + (\MP_n-\MP)[\hat{\varphi}_h(\bO; z_0)-\varphi_h(\bO;z_0)] \\
        &\,+ \MP[\hat{\varphi}_h(\bO; z_0)-\varphi_h(\bO;z_0)]+ {\theta}_h(z_0)-\theta(z_0)
    \end{aligned}
    \]
    By Proposition \ref{prop:smooth-error}--\ref{prop:bound-variance}, we have
    \[
    {\theta}_h(z_0)-\theta(z_0) = O(h^{\gamma-1})
    \]
    \[
    \begin{aligned}
      &\,\MP[\hat{\varphi}_h(\bO; z_0)-\varphi_h(\bO;z_0)]\\
      = &\, O_{\MP} \left( \frac{1}{h}\sup_{|z-z_0|\leq h} \sqrt{\ME_\bX\ME_D[(\hat{\pi}(z\mid \bX)-\pi(z\mid \bX))^2]} \sup_{|z-z_0|\leq h} \sqrt{\ME_\bX\ME_D[(\hat{\mu}( \bX,z)-\mu(\bX,z))^2]} \right).  
    \end{aligned}
    \]
    The following lemma bounds the empirical process term $(\MP_n-\MP)[\hat{\varphi}_h(\bO; z_0)-\varphi_h(\bO;z_0)]$.
\begin{lemma}\label{lemma:bounds-remainder}
    Assume we estimate nuisance functions $\pi, \mu$ from a separate independent sample, and the nuisance functions and their estimates satisfy $\epsilon \leq \pi, \hat{\pi} \leq C, |Y|, |\mu|\leq C$. Further assume the kernel $K$ satisfies $\int |K'(u)|du, \int (K'(u))^2du < \infty$. Then we have 
    \[
    \begin{aligned}
        &\,(\MP_n-\MP)[\hat{\varphi}_h(\bO; z_0)-\varphi_h(\bO;z_0)] \\
        =&\, O_{\MP} \left(  \frac{1}{\sqrt{nh^3}} \max \left\{ \sup_{|z-z_0|\leq h} \sqrt{\ME_{\bX}[\ME_D \left(\hat{\mu}(\bX,z) - \mu(\bX,z)\right)^2]}, \sup_{|z-z_0|\leq h} \sqrt{\ME_{\bX}[\ME_D \left(\hat{\pi}(z \mid \bX) - \pi(z \mid \bX)\right)^2]}\right\} \right)
    \end{aligned}
    \]
    \end{lemma}
    The asymptotic expansion follows from combining these results. To show the asymptotic normality, we need the following lemma as a sufficient condition for Lindeberg's theorem. 
    \begin{lemma}[A sufficient condition for Lindeberg's condition]\label{lemma:lindeberg}
    Suppose $\{X_{nk},n \geq 1, 1\leq k \leq k_n\}$ is a triangular array such that for each $n$, $X_{n1}, \dots, X_{nk_n}$ are independent. Let $B_n^2 = \sum_{k=1}^{k_n} \operatorname{Var}(X_{nk})$. Further assume there exists a sequence $\{L_n,n \geq 1\}$ satisfying
    \[
    \max_{1\leq k \leq k_n} |X_{nk}| \leq L_n, \,L_n/B_n \rightarrow 0.
    \]
    Then Lindeberg's condition holds, i.e., for any $\tau >0$ we have
    \[
    \lim _{n \rightarrow \infty} \frac{1}{B_n^2} \sum_{k=1}^{k_n} \mathbf{E}\left[\left(X_{nk}-\ME[X_{nk}]\right)^2 I\left(\left|X_{nk}-\ME[X_{nk}]\right| \geq \tau B_n\right)\right]=0.
    \]
    As a consequence,
    \[
    \frac{\sum_{k=1}^{k_n} (X_{nk}-\ME[X_{nk}])}{B_n} \stackrel{d}{\rightarrow} N(0,1).
    \]
    \end{lemma}
    We verify Lemma \ref{lemma:lindeberg} with $k_n=n, $
    \[
    X_{nk} = -K_h'(Z_k-z_0)  \frac{Y_k-\mu(\bX_k,Z_k)}{\pi(Z_k\mid \bX_k)} - \int  \mu(\bX_k, z)K_h'(z-z_0) dz.
    \]
    It is easy to see 
    \[
    |X_{nk}| \lesssim \frac{1}{h^2}.
    \]
    By the same logic in the proof of Proposition \ref{prop:bound-variance} one can show
    \[
    \operatorname{Var}(X_{nk}) = O\left(\frac{1}{h^3}\right).
    \]
    We further argue that 
    \[
    \operatorname{Var}(X_{nk}) = \Omega\left(\frac{1}{h^3}\right).
    \]
    Since the two terms in $X_{nk}$ are uncorrelated, we have
    \[
    \begin{aligned}
        \operatorname{Var}(X_{nk}) \geq &\, \operatorname{Var} \left( K_h'(Z-z_0)  \frac{Y-\mu(\bX,Z)}{\pi(Z\mid \bX)}\right)\\
    = &\, \ME \left[ \left(K_h'(Z-z_0) \right)^2 \frac{(Y-\mu(\bX,Z))^2}{\pi^2(Z\mid \bX)}\right]\\
    = &\,  \ME \left[ \left(K_h'(Z-z_0) \right)^2 \frac{\operatorname{Var}(Y \mid \bX, Z)}{\pi^2(Z\mid \bX)}\right]\\
    \gtrsim &\, \ME \left[ \left(K_h'(Z-z_0) \right)^2 \frac{1}{\pi^2(Z\mid \bX)}\right] \\
    = &\, \ME\left[ \int\left(K_h'(z-z_0) \right)^2\frac{1}{\pi(z\mid \bX)} dz \right] \\
    \gtrsim &\, \int\left(K_h'(z-z_0) \right)^2dz \\
    = &\, \frac{1}{h^3} \int  \left( K'(u)\right)^2 du,
    \end{aligned}
    \]
    where we use the condition $\operatorname{Var}(Y \mid \bX, Z) \geq c>0$ and $\pi \leq C$. Thus we have
    \[
    B_n^2 = \sum_{k=1}^{n} \operatorname{Var}(X_{nk}) \asymp \frac{n}{h^3}.
    \]
    Under the assumed scaling condition, we have
    \[
    L_n \asymp \frac{1}{h^2},\, B_n \asymp \frac{\sqrt{n}}{\sqrt{h^3}},
    \]
    \[
    L_n/B_n \asymp \frac{1}{\sqrt{nh}} \rightarrow 0.
    \]
    So the condition in Lemma \ref{lemma:lindeberg} holds and Lindeberg's condition holds, which further implies the asymptotic normality of $(\MP_n-\MP) [\varphi_h(\bO;z_0)]$. The remainder terms are asymptotically negligible under the rate assumptions in the theorem.
\end{proof}

\subsection{Proof of Proposition \ref{prop:IF-pseudorisk}}
\begin{proof}
    We let $\bV = \emptyset$ and $A=Z$ in Theorem 4 of \cite{kennedy2019robust}, which reduces the local IV curve to the derivative of the dose-response function, i.e., $\gamma(t) = \theta(t)$. The influence function is then given by 
    \[
    L_w(\bO) = 2\int \frac{d}{d z} \{ w(z) \bar{\theta}(z)\} \mu(\bX,z)dz - \int \frac{d}{d z} \{ w(z) \bar{\theta}^2(z)\} z dz + 2\left.\frac{d}{d z} \{ w(z) \bar{\theta}(z)\} \right|_{z=Z}\frac{Y-\mu(\bX,Z)}{\pi(Z\mid \bX)}
    \]
    Integration by part then yields
    \[
    \int \frac{d}{d z} \{ w(z) \bar{\theta}^2(z)\} z dz =- \int w(z) \bar{\theta}^2(z) dz.
    \]
\end{proof}

\end{document}